\PassOptionsToPackage{hyphens}{url}
\documentclass[pageno]{jpaper}

\usepackage[normalem]{ulem}
\usepackage[utf8]{inputenc}
\usepackage[english]{babel}
\usepackage{amsthm}
\usepackage{amsmath}
\usepackage{amssymb}
\usepackage{subfig}
\usepackage[cache=true, frozencache]{minted}
\usepackage{multicol}
\usepackage{cite}
\usepackage[algo2e, vlined]{algorithm2e}
\usepackage{authblk}

\def\pcode{\texttt}
\def\code{\pcode}

\def\dali{{Dal\'{\i}}}

\usepackage{color}

\definecolor{code_bg}{rgb}{0.98,0.98,0.98}
\definecolor{code_hl}{rgb}{0.92,0.92,0.92}

\makeatletter
\newcommand{\removelatexerror}{\let\@latex@error\@gobble}
\let\old@printtopmatter\@printtopmatter
\def\@printtopmatter{%
  \global\setbox\mktitle@bx=\vbox{\noindent\box\mktitle@bx\par\bigskip}%
  \old@printtopmatter}
\makeatother

\SetCommentSty{mycommfont}
\LinesNumbered
\DontPrintSemicolon

\def\Montage{Montage}
\def\lb{\linebreak[1]}

\newif\ifverbose \verbosetrue

\begin{document}

\title{
  Montage: A General System for Buffered Durably Linearizable Data
  Structures\thanks{%
This work was supported in part by NSF grants CCF-1422649, CCF-1717712, and
CNS-1900803, by a Google Faculty Research award, and by a US Department of
Energy Computational Science Graduate Fellowship (grant DE-SC0020347).}}

\author{Haosen Wen\thanks{The first two authors contributed equally to
    this work.}\,\strut}
\newcommand\CoAuthorMark{\footnotemark[\value{footnote}]\,\strut}
\author{Wentao Cai\protect\CoAuthorMark}
\author{Mingzhe Du}
\author{Louis Jenkins}
\author{Benjamin Valpey}
\author{Michael L. Scott}
\affil{%
\normalsize {University of Rochester \\
  Rochester, NY, USA \\
  \{hwen5,wcai6,mdu5,ljenkin4,bvalpey,scott\}@cs.rochester.edu}
}

\date{}
\maketitle

\thispagestyle{empty}

\begin{abstract}

\rule{0pt}{14pt}  
The recent emergence of fast, dense, nonvolatile main memory suggests
that certain long-lived data might remain in its natural
pointer-rich format across program runs and hardware reboots.
Operations on such data must be instrumented with explicit
write-back and fence instructions to ensure consistency in the
wake of a crash.  Techniques to minimize the cost of this
instrumentation are an active topic of research.

We present what we believe to be the first general-purpose approach to
building \emph{buffered durably linearizable} persistent data
structures, and a system, \Montage,
to support that approach.
\Montage{} is built on top of the Ralloc nonblocking persistent
allocator.  It employs a slow-ticking \emph{epoch clock}, and ensures
that no operation appears to span an epoch boundary.  It also arranges
to persist only that data minimally required to reconstruct the
structure after a crash.  If a crash occurs in epoch $e$, all work
performed in epochs $e$ and $e-1$ is lost, but work from prior
epochs is preserved.

We describe the implementation of \Montage, argue its correctness, and
report unprecedented throughput for persistent queues, sets/mappings,
and general graphs.

\end{abstract}

\section{Introduction}
\label{sec:introduction}

Despite enormous increases in capacity over the years,
the dichotomy between transient working memory (DRAM) and
persistent long-term storage (magnetic disks and flash) has been a
remarkably stable feature of computer organization.  Finally, however,
DRAM is approaching end-of-life.  Successor technologies will be denser
and much less power hungry.  They will also be nonvolatile.
\ifverbose
Already, today, one can buy an Intel server with multiple terabytes
of byte-addressable phase-change memory for less than \$20K\,USD.
\fi
While it is entirely possible to use nonvolatile memory (NVM) as a
plug-in replacement for
DRAM, nonvolatility raises the intriguing possibility of keeping
pointer-rich data ``in memory'' across program runs and even system
crashes, rather than serializing it to and from a file system or
back-end database.

Crashes cause problems, however.  For file systems and databases,
long-established logging techniques ensure that transitions
from one consistent state to another are failure atomic.  For data
structures accessed with load and store instructions, the cost
of such logging may be prohibitively high.  Moreover, the fact that
caches remain volatile and may write back their contents out of
program order means that data structure operations must typically issue
explicit write-back and fence instructions to guarantee post-crash consistency.

Past work has established \emph{durable linearizability} as the standard
correctness criterion for persistent data
structures~\cite{izraelevitz-disc-2016}.  This criterion builds on the
familiar notion of linearizability for concurrent (non-persistent) data
structures.  A data structure is said to be linearizable if whenever
threads perform operations concurrently, the effect is as if the
operations had been performed sequentially in some order that is
consistent with \emph{real time}
order (if operation $A$ returns before operation $B$ is called, then $A$ must
appear to happen before $B$) and with the semantics of the
abstraction represented by the structure.

A persistent data structure is said to be \emph{durably linearizable} if
(1) it is linearizable during crash-free operation, (2) each operation
persists (reaches a state that will survive a crash) between its call
and return, and (3) the order of persists matches the linearization
order.  By introducing program-wide coordination of persistence,
\emph{buffered} durably linearizable data structures may reduce
ongoing overhead by preserving only some consistent prefix of the
history prior to a crash.

Recent publications have described many individual
durably linearizable data structures and perhaps two dozen
general-purpose systems to provide failure atom\-i\-ci\-ty for outermost
critical sections or speculative transactions
(Sec.~\ref{sec:related}).
The need for operations to persist before returning is
a significant source of overhead in these systems.
To reduce
this overhead, Nawab et al.\ developed a buffered
durably linearizable 
hash table (\dali~\cite{nawab-disc-2017}) that ensures persistence
on a \emph{periodic}
\ifverbose (as opposed to incremental) \fi
basis.
More recently, Haria et al.'s MOD project~\cite{haria-asplos-2020} proposed that
programmers rely on history-preserving (``functional'') tree structures,
in which each update can be persisted by updating a single root pointer,
eliminating the need for logging.  Memaripour et
al.'s Pronto project~\cite{memaripour-asplos-2020} proposed that
concurrent objects log their \emph{high level} (abstract)
operations\ifverbose { (rather than low-level updates)}\fi, together with occasional checkpoints;
on a crash, they replay the portion of the log that follows the
most recent checkpoint.

Inspired in part by these previous projects, we present what we believe
to be the first general-purpose approach to \emph{buffered} durably
linearizable structures.  Our system, \Montage, employs a slow-running
\emph{epoch clock}, and ensures that no operation appears to span an
epoch boundary.  If a crash occurs in epoch $e$, \Montage{} recovers the
state of the abstraction from the end of epoch $e-2$ and rebuilds the
concrete structure.

\Montage{} also distinguishes between the \emph{abstract} state of the
concurrent object and its \emph{concrete} state.  It encourages the programmer
to maintain only the former in NVM, to reduce persistence
overhead.  A \Montage{} mapping, for example, would typically persist only a bag
of key-value pairs; the look-up structure (hash table, tree, skip list) lives
entirely in transient DRAM\@.  During recovery, \Montage{} cooperates with the
user program to rebuild the concrete state.

Our implementation of \Montage{} is built on top of
Ralloc~\cite{cai-ismm-2020}, a lock-free allocator for persistent
memory.  \Montage{} itself is also lock-free during normal operation,
though a stalled thread can arbitrarily delay progression of the
persistence frontier.
\ifverbose
We have designed an extension to \Montage{}
that avoids even this more limited form of blocking; given that
threads in real systems are seldom preempted for more than a fraction of
a second, however, the complexity of the extension seems unwarranted in
practice.
\fi

Performance experiments (Sec.~\ref{sec:results}) reveal that a \Montage{}
hash map can sustain well over 20\,M ops/s on a
read-heavy workload---7$\times$ as many as Dal\'{\i}, 17$\times$ as many
as Pronto, and within a factor of 3 of a transient DRAM table. This is
close to the best one could hope for: read latency for Intel Optane NVM
is about 3$\times$ that of DRAM~\cite{izraelevitz-optane-2019}.

\ifverbose
After reviewing related work in Section~\ref{sec:related}, we provide a
high-level description of \Montage{} in Section~\ref{sec:design}.
Correctness arguments appear in Section~\ref{sec:correctness};
implementation details appear in Section~\ref{sec:implementation}.
Section~\ref{sec:results} presents performance results, including an
exploration of the \Montage{} design space and a comparison to competing
systems.  Section~\ref{sec:conclusions} presents conclusions.
\fi

\section{Related Work}
\label{sec:related}

The past few
years have seen an explosion of work on persistent data structures, much
of it focused on B-trees indices for file systems and
databases~\cite{venkataraman-fast-2011,
  chen-vldb-2015, yang-fast-2015, oukid-sigmod-2016, kim-tos-2018,
  hwang-fast-2018, liu-icpp-2019, nam-tos-2020}.
Other work has targeted
RB trees~\cite{wang-tos-2018},
radix trees~\cite{lee-fast-2017},
hash maps~\cite{schwalb-imdm-2015, nawab-disc-2017, zuriel-oopsla-2019},
and queues~\cite{friedman-ppopp-2018}.
Several projects persist only parts of a data
structure, and rebuild the rest on recovery.  Zuriel et
al.~\cite{zuriel-oopsla-2019} argue that this approach can be used for
almost any implementation of a set or mapping.  Unfortunately, their
technique 
keeps a full copy of the
structure in DRAM, forfeiting the \ifverbose much larger \else high \fi
capacity of NVM\@.
\Montage{} eliminates this restriction; it also accommodates not only sets
and mappings, but any abstraction that comprises items and
relationships---effectively, anything that can be represented as a graph.

Several existing data structures are designed to linearize by using a
single compare-and-swap (CAS) instruction to replace a portion of the
structure~\cite{chen-vldb-2015, nawab-disc-2017, lee-fast-2017,
  nam-tos-2020}.  If the new portion is persisted before the CAS, and
the updated pointer is persisted immediately after the CAS, no separate
logging is required.  Mahapatra et al.~\cite{mahapatra-2019} and Haria
et al.~\cite{haria-asplos-2020} apply this observation
to a variety of ``functional'' data structures, building
sets, maps, stacks, queues, and vectors.  As an extension, a
sequence of single-CAS steps can be used to move a structure through
self-documenting intermediate stages~\cite{hwang-fast-2018,
  wang-tos-2018}.  In a similar vein, hardware transactional memory can
be used to modify a data structure and a log
concurrently~\cite{liu-icpp-2019}, or to update an entire cache line
without any chance that an intermediate version will be written back to
memory~\cite{kim-tos-2018}.

Izraelevitz et al.~\cite{izraelevitz-disc-2016}
provide a mechanical construction to convert any 
nonblocking concurrent structure into a correct persistent version.
David et al.~\cite{david-atc-2018} describe several techniques to
eliminate redundant writes-back and fences for such structures,
significantly improving performance.
 
\ifverbose Beyond individual data structures, several
\else Several
\fi groups have developed systems
to ensure the failure atomicity of lock-based
critical sections~\cite{hsu-eurosys-2017, chakrabarti-oopsla-2014,
  izraelevitz-asplos-2016, liu-micro-2018} or
speculative transactions~\cite{volos-asplos-2011, coburn-asplos-2011,
  charzistergiou-vldb-2015, giles-msst-2015, correia-spaa-2018,
  cohen-oopsla-2018, memaripour-iccd-2018, ramalhete-dsn-2019,
  gu-atc-2019, beadle-qstm-2020, pavlovic-podc-2018, intel-pmdk}.
Significantly, \emph{all} of these systems ensure that an operation has
persisted before permitting the calling thread to proceed---that is,
they adopt the strict version of durable linearizability.

The \dali\ hash map~\cite{nawab-disc-2017}
delays persistence, so the overhead of writes-back and fencing
can be amortized over many operations while still providing
\emph{buffered} durable linearizability.
The implementation relies on a flush-the-whole-cache
instruction, available only in privileged mode on the x86, and with
the side effect of unnecessarily evicting many useful lines.  Our
reimplementation of \dali\ (used in Sec.~\ref{sec:results}) tracks
to-be-written-back lines explicitly in software---as does \Montage.
\Montage{} then extends delayed persistence to arbitrary data structures.

Perhaps the closest prior work to \Montage, in motivation and generality,
is the Pronto system of Memaripour et al.~\cite{memaripour-asplos-2020}.
As noted in Section~\ref{sec:introduction}, Pronto logs \emph{high level}
(abstract) operations rather than low-level updates, and replays the log
after a crash.
\ifverbose
Periodic checkpoints allow it to bound the length of the
log, and thus recovery time.
\fi
Notably, Pronto still pays the cost of
persisting each operation before returning.

\ifverbose
\Montage's use of a global epoch clock has several precedents, including
implementations of software transactional memory~\cite{dice-disc-2006,
  riegel-disc-2006} and of safe memory reclamation for
transient~\cite{fraser-thesis-2004} and persistent~\cite{david-atc-2018}
data structures.
\fi

\section{\Montage{} Design}
\label{sec:design}

\Montage{} manages persistent \emph{payload} blocks on behalf of one or
more concurrent data structures.  A programmer who wishes to adapt a
structure to \Montage{} must identify the subset of the structure's data
that is needed, in quiescence, to capture the state of the abstraction.
A set, for example, needs to keep its items in payload blocks, but not
its lookup structure.  A mapping needs to keep key-value pairs.  A queue
needs to keep its items \emph{and} their order: it might label
payloads with consecutive integers from $i$ (the head) to $j$ (the tail).
A graph can keep a payload for each vertex (each with a unique name) and
a payload for each edge (each of which names two vertices).

A typical data structure maintains additional, transient data to speed
up operations.  A set or mapping might maintain a hash
table, tree, or skip list as an index into the pile of items or pairs.
A queue might maintain a linked list of pointers to items.  A graph
(depending on the nature of the application) might maintain a transient
object for each vertex, containing a pointer to a payload for the vertex
attributes,
a set of pointers to neighboring vertex objects, and (if edges have
large attributes) a set of pointers to edge payloads.
All of this transient data must be reconstructed after a crash.

Crucially, synchronization is always performed on transient data.  That
is, \emph{\Montage{} does not determine the linearization order for
operations on a data structure.}  Rather it ensures that the
persistence order for payloads is consistent with the linearization
order provided by the transient structure.
More specifically, it divides execution into \emph{epochs} in such a way
that every epoch boundary represents a consistent cut of the
happens-before relationship among operations; it then arranges, in the
wake of a crash, to recover all managed data structures to their state
as of some common epoch boundary.

\subsection{API}
\label{sec:api}

The \Montage{} API for C++ is shown in Figure~\ref{fig:api}.
An example of a lock-based hash table
is shown in Figure~\ref{fig:api-example-hashtable}.

\begin{figure}
\begin{minted}[fontsize=\footnotesize,bgcolor=code_bg]{c++}
namespace pds{
/* Payload class infrastructures */
// Base class of all Payload classes
class PBlk;
// Macro to generate get() and set() methods for
// fields of payload_type, where payload_type is a
// subtype of PBlk and fieldname has type type_name
GENERATE_FIELD(type_name, fieldname, payload_type);
// The macro expands to a protected 
// field `m_fieldname` and the following members:
  // get value with old-see-new alert enabled
  const type_name& get_fieldname();
  // get with old-see-new alert disabled
  const type_name& get_unsafe_fieldname();
  // set value of fieldname with new value
  // may return a copy of new payload
  payload_type* set_fieldname(type_name&);

/* Methods */
// Consistently begin operation in current epoch
// Optionally take PBlk's as arguments and
// mark them with current epoch
void BEGIN_OP(optional<PBlk*>, ...);
// End an operation
void END_OP();
// Begin an operation that will run through
// end of the scope, using the RAII idiom
BEGIN_OP_AUTOEND(optional<PBlk*>, ...);
// Create a payload block
payload_type* PNEW(payload_type, ...);
// Delete a payload after the end of next epoch
void PDELETE(PBlk*);
// Throw exception if the epoch has changed
CHECK_EPOCH();
// Mark a payload deleted but
// hold it until PRECLAIM is called
void PRETIRE(PBlk*);
// Delete a PRETIRE-d payload as soon as possible
void PRECLAIM(PBlk*);

/* Old-see-new alert */
struct OldSeeNewException : public std::exception;
};
\end{minted}
\caption{C++ API.}
\label{fig:api}
\end{figure}

\begin{figure}

\begin{minted}[xleftmargin=12pt,numbersep=4pt,linenos,
fontsize=\footnotesize,bgcolor=code_bg,
highlightlines={3,4,5,14,17,20,29},
highlightcolor=code_hl]{c++}
class HashTable{
// Payload class
class Payload : public PBlk{
  GENERATE_FIELD(K, key, Payload);
  GENERATE_FIELD(V, val, Payload);
}
// Transient index class
struct ListNode{
  // Transient-to-persistent pointer
  Payload* payload = nullptr;
  // Transient-to-transient pointers
  ListNode* next = nullptr;
  void set_val_wrapper(V& v){
    payload = payload->set_val(v);
  }
  ListNode(K& key, V& val){
    payload = PNEW(Payload, key, val);
  }
  ~ListNode(){
      PDELETE(payload);
  }
  // get() methods omitted
}
// Insert, or update if the key exists
optional<V> put(K key, V val, int tid){
  size_t idx=hash_fn(key)%idxSize;
  ListNode* new_node = new ListNode(key, val);
  std::lock_guard lk(buckets[idx].lock);
  BEGIN_OP_AUTOEND(new_node->payload);
  ListNode* curr = buckets[idx].head.next;
  ListNode* prev = &buckets[idx].head;
  while(curr){
    K& curr_key = curr->get_key();
    if (curr_key == key){
      optional<V&> ret = curr->get_val();
      curr->set_val_wrapper(val);
      delete new_node;
      return ret;
    } else if (curr_key > key){
      new_node->next = curr;
      prev->next = new_node;
      return {};
    } else {
      prev = curr;
      curr = curr->next;
    }
  } // while
  prev->next = new_node;
  return {};
}
};
\end{minted}
\captionsetup{justification=centering}
\vspace{-3ex}
\caption{Simple lock-based hash table example\\
(\Montage-related parts highlighted).}
\label{fig:api-example-hashtable}
\end{figure}

Any operation that creates or updates payloads must make itself visible
to \Montage{} by calling \code{BEGIN\_OP} or
\code{BEGIN\_\lb{}OP\_\lb{}AUTOEND}.
It indicates completion with \code{END\_OP}.
Read-only operations can skip these calls, though
they must still synchronize on the transient data structure.
Payloads are created and destroyed using \code{PNEW} and
\code{PDELETE}.
\code{PRETIRE} and \code{PRECLAIM} take the place of
\code{PDELETE} for nonblocking
memory management (Sec.~\ref{sec:nonblocking-datastructures}).
Existing payloads are accessed with
\code{get} and \code{set} methods, created by the
\code{GENERATE\_FIELD} macro; \code{get} returns a \code{const}
reference to the field; \code{set} updates the field and returns a
(possibly altered) pointer to the payload\ifverbose { as a whole}\fi.

To support the epoch system, \Montage{} labels all payloads with the
epoch in which they were created \emph{or most recently modified}.
An operation in epoch $e$ that wishes to modify an existing
payload can do so ``in place'' if the payload was created in $e$;
otherwise, \Montage{} creates a \emph{new} payload with which to replace it.
The \code{set} methods enforce this convention by returning a pointer to
a new or copied payload, as appropriate.

Because epochs are long (10--100\,ms),
``hot'' payloads are typically modified in place.  When
a new copy is created, however, an operation must re-write any
pointers to the payload found anywhere in the structure.
For this reason, it is important to minimize the number of pointers to a
given payload found in transient data.  It is even more important to
avoid long chains of pointers in persistent data: otherwise, a change to
payload $p$, at the end of a long chain, would require a change
to the penultimate payload $p^\prime$, which would in
turn require a change to its predecessor $p^{\prime\prime}$,
and so on.
\ifverbose
A similar observation is made by the designers of
MOD~\cite{haria-asplos-2020}.
\fi

Because calls to \code{get} are invisible to recovery, they can safely
be made outside the bounds of \code{BEGIN\_OP} and \code{END\_OP}
(subject 
to transient synchronization).
Calls to \code{PNEW} can also be made early, so long as the payloads
they return are passed as parameters to \code{BEGIN\_OP}, so they can
be properly labeled.

\subsection{Periodic Persistence}
\label{sec:epochs}

The key task of \Montage{} is to ensure that operations persist in an order
consistent with their linearization order.  Toward that end, the system
ensures that
\begin{enumerate}
\item\label{prop:same-epoch}
  all payloads created or modified by a given operation are labeled
  with the same epoch number;
\item\label{prop:persist-together}
  all payloads created or modified in a given epoch $e$ persist
  together, instantaneously, when the epoch clock ticks over from
  $e+1$ to $e+2$;
and
\item\label{prop:linearize-in-epoch}
  each update operation linearizes in the epoch in which it created
  payloads.
\end{enumerate}

\smallskip\noindent
Property \ref{prop:same-epoch} is ensured by the \code{set} and
\code{PNEW} methods, as described in Section~\ref{sec:api}.  Note that
an operation that begins in epoch $e$ can continue to create and modify
payloads in that epoch, even if the clock ticks over to $e+1$.

Property \ref{prop:persist-together} is enforced by \Montage's recovery
routines: if a crash occurs in epoch $e$, those routines discard all
payloads labeled $e$ or $e-1$, but keep everything that is older.  Note
that this convention requires that deletion be delayed.  If a payload
created or updated in epoch $b$ is passed to \code{PDELETE} in epoch
$e > b$, the \code{PDELETE} method creates an ``anti-payload'' labeled~$e$.
If a crash occurs before $e+2$, the anti-payload will be discarded
and the original payload retained.  If a crash occurs immediately after
the tick from $e+1$ to $e+2$, the anti-payload will be discovered during
recovery and both it and the original payload will be discarded.  If
execution proceeds without a crash, the original payload will be
reclaimed when the epoch advances from $e+2$ to $e+3$; the anti-payload
will be reclaimed when the epoch advances from $e+3$ to $e+4$.

Property \ref{prop:linearize-in-epoch} is the responsibility of the
transient data structure built on top of \Montage.  Lock-based
operations are easy: no conflicting operation can proceed until we
release our locks, and we can easily pretend that all updates happened
at the last call to \code{set} or \code{PNEW}.  For nonblocking
structures, a similar guarantee can be made if every operation
linearizes on a statically identified compare-and-swap (CAS)
instruction that also modifies an adjacent counter (as is often used to
avoid ABA anomalies).
One first reads some variable $x$, then double-checks the epoch
clock (the \code{CHECK\_EPOCH} method exists for this purpose), and only
then attempts a CAS on $x$\@.  If the CAS succeeds, it can be said to
have occurred at the time of the \code{CHECK\_EPOCH} call.
Note that this strategy generally requires read-only operations on the
same structure to be modified by replacing their linearizing read with a
CAS that updates the adjacent count: otherwise a read that occurs
immediately after an epoch change might observe an update from the
previous epoch as not yet having occurred.  For cases in which this
modification is undesirable (e.g., because reads vastly outnumber
updates), we use a variant of the double-compare-single-swap (DCSS)
software primitive of Harris et al.~\cite{harris-disc-2002}) to update
a location while simultaneously verifying the current epoch number.  A
compatible read primitive performs no store instructions (and thus
induces no cache evictions) so long as no DCSS is currently in progress
in another thread (if one is, the read helps the DCSS complete).

As an assist to programmers in ensuring
property~\ref{prop:linearize-in-epoch}, \Montage{} raises an exception
whenever an operation running in epoch $e$ reads a payload
created in some epoch $e' > e$.
In most cases, programmers can ensure
that this exception will never arise.
In other cases, the operation may respond to the exception by rolling
back what has done so far and starting over in the newer epoch.
In special cases, an operation can ignore the exception
or use \code{get\_unsafe} methods to avoid generating it in the first
place (the new data might, for example, be used only for semantically
neutral performance enhancement).

\smallskip
In support of these properties,
the epoch-advancing mechanism at
the end of epoch $e$
\begin{itemize}
\item
  waits until no operation is active in epoch $e-1$;
\item
  reclaims all payloads deleted in epoch $e-2$ and all anti-payloads
  created in epoch $e-3$;
\item
  explicitly writes back all payloads created or modified in epoch $e-1$;
\item
  waits for the writes-back to complete;
and
\item
  updates and writes back the epoch clock.
\end{itemize}

\noindent
Further details appear in Section~\ref{sec:implementation}.
\subsection{Nonblocking Data Structures}
\label{sec:nonblocking-datastructures}

As described in Section~\ref{sec:epochs}, \Montage{} is compatible with
nonblocking operations that employ special CAS or read primitives to
ensure that linearization occurs in the epoch in which any payloads were
created or modified.

In the general case, a structure that uses the
\code{Old\-See\-New\-Excep\-tion} to keep its linearization order
consistent with epoch order
may find that the resulting restarts make it lock-free or
obstruction-free, rather than wait-free.
Still, nothing in \Montage{} precludes lock freedom.
\ifverbose
At the same time, while \Montage{} never indefinitely delays an operation
unless some other operation has made progress,
a stalled operation can indefinitely delay the progress of persistence.

\ifverbose
We have designed (but not yet implemented) a version of \Montage{} that
allows the epoch clock (and thus the persistence frontier) to
advance in a nonblocking fashion.
\Montage{} already maintains, internally, an array that indicates, for each
thread, whether that thread is actively executing a data structure
operation, and, if so, in which epoch.  (It is by scanning this array
that the epoch-advancing mechanism knows whether it can proceed.)  The
key to nonblocking persistence is to augment this array with a ``serial
number'' for the current operation, and an indication of whether that
operation is active, committed, or aborted, much as nonblocking
object-based software transactional memory systems track the status of
transactions~\cite{herlihy-podc-2003, marathe-disc-2005,
fraser-tocs-2007, marathe-transact-2006, tabba-spaa-2009}.
Each payload is then labeled not only with an epoch number, but with the
thread id and serial number of its creator.  To advance the epoch, we
scan the array and abort any operation that stands in our way by
CAS-ing its status from ``active'' to ``aborted.''  Each data structure
operation, for its part, ends by CAS-ing its status from ``active'' to
``committed'' and, if the epoch has advanced, performing a write-back
and fencing that status.  Recovery routines, in the wake of a crash,
discard any payloads---even in old epochs---whose creating operations
were aborted.
\fi

\ifverbose 

Nonblocking structures may need to use \emph{safe memory reclamation}
(SMR) techniques, such as epoch-based
reclamation~\cite{fraser-thesis-2004} or hazard
pointers~\cite{michael-tpds-2004}, to avoid creating dangling pointers
when deleting blocks to which other threads may still hold references.
In such structures, the actual reclamation of a block may be occur
outside the scope---or even the
epoch---of its deleting operation.  Simply
delaying the \code{PDELETE} of a payload to the
reclamation of some corresponding transient block does not suffice,
because transient \emph{limbo lists} belonging to SMR are lost after a
crash.
\Montage{} provides \code{PRETIRE} and \code{PRECLAIM} operations to
handle this situation.

\else
Most nonblocking structures will use \code{PRETIRE} and
\code{PRECLAIM} instead of \code{PDELETE}.  Similar routines (e.g.,
for epoch-based reclamation~\cite{fraser-thesis-2004} or hazard
pointers~\cite{michael-tpds-2004}) are assumed to be used for
transient data, but are not provided by \Montage.
\fi

Typically, \code{PRETIRE} is called when a
payload is ``detached'' from the shared structure, and \code{PRECLAIM} is called
upon the destruction of its transient parent, when no references to the
payload remain outside the memory
manager.  After a crash in epoch $e$, a payload \code{PRETIRE}-d in or
before $e-2$ but not \code{PRECLAIM}-ed can safely be reclaimed.
\fi

\section{Correctness}
\label{sec:correctness}
\newtheorem{theorem}{Theorem}
\newtheorem{lemma}{Lemma}

We argue that \Montage{} (1) preserves the linearizability of a structure
implemented on top of it, (2) adds buffered durable linearizability, and
(3) preserves lock freedom.

Each concurrent data structure serves to implement some abstract data type.
The semantics of such a type are defined in terms of
\emph{legal histories}---sequences of operations, with their arguments
and return values.  The implementation is correct if it is
\emph{linearizable}, meaning that every concurrent history (with
overlapping calls and returns from different threads) is equivalent to
(has the same operations as) some sequential history that is consistent
with real-time order (if $a$ returns before $b$ is called in the
concurrent history, then $a$ precedes $b$ in the sequential history) and
that represents a valid operation sequence for the data type.

We can define the abstract \emph{state} of a data type, after a finite
sequence of operations, as the set of sequences that are permitted to
extend that sequence according to the type's semantics.  Suppose, then,
that data structure $S$ is a correct implementation of data type $T$,
and that $s$ is a quiescent concrete state of $S$ (the bits in memory at
some point when no operations are active).  We can define the
\emph{meaning} of that state, ${\cal M}(s)$, as the state of $T$ after
the sequence of abstract operations corresponding to (a linearization
of) the operations performed so far on $S$.

We assume that the programmer using \Montage{} obeys the following
\emph{well-formedness} constraints:
\begin{enumerate}
\item\label{constraint:naked-lin}
    Each data structure $S$, implemented on top of \Montage, is
    linearizable when \Montage{} itself is disabled and crashes do not occur.
    More specifically, assume that
        (a) \code{PNEW} and \code{PDELETE} are
        implemented as ordinary \code{new} and \code{delete};
        (b) \code{get} and \code{set} are ordinary accessor methods, and
        \code{set} never copies a payload;
        (c) \code{BEGIN\_OP} and \code{END\_OP} are no-ops; and
        (d) the \code{OldSeeNewException} never arises.
    Under these circumstances, the structure is linearizable.
\item\label{constraint:drf-payloads}
    Any synchronization required for linearizability is performed solely on
    transient data---accesses to payloads never participate in a data
    \emph{or} synchronization race.
\item\label{constraint:epoch-bracketing}
    All accesses to payloads are made through \code{get} and \code{set}.
    Each operation that modifies the data structure (a) calls
    \code{BEGIN\_OP} before \code{set} (passing as arguments any
    previously created payloads), (b) calls \code{END\_OP} after
    completing all its \code{set}s, and (c) ensures that between its
    last call to \code{set} or \code{CHECK\_EPOCH} and its linearization
    point, no conflicting operation can linearize.
\item\label{constraint:repatching}
    Whenever \code{set} returns a pointer to a payload different than the
    one on which it was called, the calling operation replaces every pointer
    to the old payload in the structure with a pointer to the new payload.
\item\label{constraint:payload-state}
    There exists a mapping $\cal Q$ from sets of payloads to states of $T$
    such that whenever $S$ is quiescent, ${\cal M}(s) = {\cal Q}(p)$, where
    $s$ is the concrete state of $S$ and $p$ is the current set of payloads.
\item\label{constraint:recovery}
    The recovery routine for $S$, given a set of payloads $r$,
    constructs a concrete state $t$ such that ${\cal M}(t) = {\cal Q}(r)$.
\end{enumerate}

\subsection{Linearizability}

\begin{lemma}\label{lemma:montage_on}
    A well-formed, linearizable concurrent data structure, implemented
    on top of \Montage, remains well-formed and linearizable when \Montage{}
    is enabled.
\end{lemma}

\begin{proof}
    Constraint~\ref{constraint:repatching} ensures that any payload
    cloned by \Montage{} is reattached to the structure wherever the old
    payload appeared.  Since access to payloads is race-free
    (Constraint~\ref{constraint:drf-payloads}), this re-attachment
    is safe.
    Throws of the \code{Old\-See\-New\-Excep\-tion} will be harmless:
    they exist simply to simplify compliance with
    Constraint~\ref{constraint:epoch-bracketing}; any operation that
    already satisfies that constraint can safely ignore the exception.
    Finally, given the mapping $\cal Q$ from payloads to abstract state
    (Constraint~\ref{constraint:payload-state}), we can easily create a
    ${\cal Q}^\prime$ that ignores both the old versions of cloned
    payloads and any payloads for which an anti-payload exists.
    These are the only effects of enabling \Montage{} that are visible to
    the structure during crash-free execution.
\end{proof}

\begin{theorem}
    A \Montage{} data structure $S$ remains linearizable when epoch advancing
    operations are added to its history.
\end{theorem}

\begin{proof}
    Let $a_e$ denote the operation that advances the epoch from $e-1$ to $e$.
    Consider a linearization order for $S$ itself, as provided by
    Lemma~\ref{lemma:montage_on}.
    Constraint~\ref{constraint:epoch-bracketing} ensures that the
    linearization point of any update operation in this order occurs
    between events $a_{e}$ and $a_{e+1}$, making it easy to place these
    events into the linearization order.
    A read-only operation, moreover, has no forward or anti-dependences
    on the epoch clock, and so cannot participate in any circular
    dependence with respect to the epoch advancing events.
\end{proof}

\subsection{Buffered Durable Linearizability}

\begin{theorem}
    A well formed, linearizable concurrent data structure, running on
    \Montage, is buffered durably linearizable.
\end{theorem}

\begin{proof}
We need to show that in any execution $H$ containing a crash $c$, the
state of the data structure after recovery reflects some consistent
prefix of the linearized pre-crash history.  Suppose that $c$ occurs
in epoch $e$ of $H$.  If $e \le 2$,
recovery will restore the initial state of the system, which reflects
the null prefix of execution.  If $e > 2$, \Montage{} will discard all
payloads created in epochs $e$ and $e-1$, preserving those in existence
as of $a_{e-1}$, and will pass these to the structure's recovery
routine.  This routine, by Constraint~\ref{constraint:recovery}, will
construct a new concrete state $t$ such that ${\cal M}(t) = {\cal
  Q}(r)$, where $r$ is the set of payloads it was given.
But $r$ is precisely the set of payloads created by operations that
linearized prior to $a_{e-1}$.  If execution had reached quiescence
immediately after those operations,
Constraint~\ref{constraint:payload-state} implies that the concrete
state $s$ of $S$ would have been such that ${\cal M}(s) = {\cal Q}(r)$.
Thus the post-recovery state $t$ reflects a consistent prefix of the
linearized pre-crash history.
\end{proof}

\subsection{Liveness}

\begin{theorem}
    \Montage{} is lock free during crash-free execution.
\end{theorem}

\begin{proof}
    The only loop in \Montage{} lies within \code{BEGIN\_OP}, where an
    update operation seeks to read the epoch clock and announce itself
    as active in that epoch, atomically.
    Each retry of the loop implies that the epoch has advanced.  If we
    assume that the epoch advancing operation (which need not be
    nonblocking) always waits until at least one operation has completed
    in the old epoch, then an operation can be delayed in
    \code{BEGIN\_OP} only if some other operation has completed.
    The \code{OldSeeNewException}, similarly, will arise (and
    cause some operations to start over) only if the epoch has advanced.
\end{proof}

\section{Implementation Details}
\label{sec:implementation}
Figure~\ref{fig:montage-pseudocode} shows pseudocode for \Montage's
core functionality.
Transient data structures include an ``operation tracker'' that
indicates, for each thread in the system, the epoch of its active
operation (if any), and lists of payloads to be persisted and freed
at future epoch boundaries.  The latter are logically indexed by epoch,
but only the most recent 2 or 3 are needed.  For simplicity, \Montage{}
maintains four sets of lists, and indexes into them using the 2
low-order bits of the epoch number.
For convenience, each thread also caches the epoch of its currently
active operation (if any) in thread-local storage.

\SetKwRepeat{Do}{do}{while}
\SetKw{Or}{or}
\SetKw{And}{and}
\SetKw{Return}{return}
\SetKw{New}{new}
\SetKw{Delete}{delete}
\SetKw{Throw}{throw}
\SetKw{Sfence}{sfence}
\SetKw{Local}{operation\_local}
\SetKwProg{Fn}{Function}{}{}
\SetKwProg{Macro}{Macro}{}{}
\SetKwProg{Class}{Class}{}{}
\SetKwProg{Struct}{Struct}{}{}

\begin{figure*}[!t]
\setlength{\columnsep}{0pt}
\removelatexerror
\begin{algorithm2e}[H]
\footnotesize
\begin{multicols}{2}
\tcp*[l]{transient structures}
$operation\_tracker$\;
$to\_be\_persisted\,$[4] \tcp*[l]{recent 4 epochs}
$to\_be\_freed\,$[4] \tcp*[l]{recent 4 epochs}
\Local $op\_epoch$\;
\tcp*[l]{persistent structures}
$curr\_epoch$\;
\Struct{Payload}{
  $type$ = \{ALLOC, UPDATE, DELETE\}\;
  $epoch$\;
  $uid$ \tcp*[l]{shared between real and anti-payloads}
}
\tcp*[l]{utility function}
\Fn{verify\,(Payload* p) : void}{
  \If{$op\_epoch$ < p$\rightarrow$$epoch$}{
    \Throw OldSeeNewException\;
  }
}
\Fn{advance\_epoch\,(\,) : void}{
  $operation\_tracker.wait\_all\,$($curr\_epoch$\,-\,1)\;
  $to\_be\_freed\,$[$curr\_epoch$\,-\,2].$free\_all$(\,)\;
  $to\_be\_persisted\,$[$curr\_epoch$\,-\,1].$persist\_all$(\,)\;
  \Sfence\;
  $curr\_epoch.atomic\_increment\,$(\,)\;
}
\Macro{BEGIN\_OP : void}{
  \Repeat{op\_epoch == curr\_epoch}{
    $op\_epoch$ = $curr\_epoch$\;
    $operation\_tracker.register\,$($tid$, $op\_epoch$)\;
  }
}
\Macro{END\_OP : void}{
  $op\_epoch$ = NULL\;
  $operation\_tracker.unregister\,$($tid$)\;
}
\Macro{PNEW\,(Type, ...) : Type*}{
  $new\_payload$ = \New $Type\,$(...)\;
  $new\_payload\rightarrow epoch$ = $op\_epoch$\;
  $new\_payload\rightarrow type$ = ALLOC\;
  \Return $new\_payload$\;
}
\Macro{PDELETE\,(Payload* p) : void}{
  $verify\,$($p$)\;
  \eIf{p.epoch == op\_epoch}{
    \eIf{p$\rightarrow$type == $ALLOC$}{
      delete($p$)\;
      \Return\;
    }{
      $p\rightarrow type$ = DELETE\;
    }
  }{
    $anti\_payload$ = \New $Payload\,$(\,)\;
    $anti\_payload\rightarrow type$ = DELETE\;
    $anti\_payload\rightarrow uid$ = $p\rightarrow uid$\;
    $to\_be\_persisted\,$[$op\_epoch$].$add\,$($anti\_payload$)\;
    $to\_be\_freed\,$[$op\_epoch$\,+\,1].$add\,$($anti\_payload$)\;
  }
  $to\_be\_freed\,$[$op\_epoch$].$add\,$($p$)\;
}
\Fn{payload.get\_x(\,) : typeof(x)}{
  $verify\,$(this)\;
  \Return this$\rightarrow x$\;
}
\Fn{payload.set\_x\,(typeof(x) y) : Payload*}{
  $verify\,$(this)\;
  \eIf{this$\rightarrow epoch$ == $op\_epoch$}{
    this$\rightarrow x$ = $y$\;
    $to\_be\_persisted\,$[$op\_epoch$].$add\,$(this)\;
    \Return this\;
  }(this$\rightarrow epoch$ < $op\_epoch$){
    $new\_payload$ = copy(this)\;
    $new\_payload\rightarrow epoch$ = $op\_epoch$\;
    $new\_payload\rightarrow type$ = UPDATE\;
    $new\_payload\rightarrow x$ = $y$\;
    $to\_be\_persisted\,$[$op\_epoch$].$add\,$($new\_payload$)\;
    $to\_be\_freed\,$[$op\_epoch$].$add\,$(this)\;
    \Return $new\_payload$\;
  }
}
\ifverbose
\BlankLine
\tcp*[l]{payload of nonblocking data structures}
\Struct{NBPayload : public Payload}{
  \tcp*[l]{transient; discarded during recovery}
  $Payload$* $anti\_payload$ = NULL\;
}
\BlankLine
\Macro{PRETIRE\,(NBPayload* p) : void}{
  $verify\,$($p$)\;
  \eIf{p$\rightarrow$epoch == op\_epoch}{
    $p\rightarrow type$ = DELETE\;
  }($p\rightarrow epoch$ < $op\_epoch$){
    $p\rightarrow anti\_payload$ = \New $Payload\,$(\,)\;
    $p\rightarrow anti\_payload\rightarrow uid$ = $p\rightarrow uid$\;
    $p\rightarrow anti\_payload\rightarrow type$ = DELETE\;
    $p\rightarrow anti\_payload\rightarrow epoch$ = $op\_epoch$\;
    $to\_be\_persisted\,$[$op\_epoch$].$add\,$($p\rightarrow anti\_payload$)\;
  }
  $to\_be\_persisted\,$[$op\_epoch$].$add\,$($p$)\;
}
\BlankLine
\Macro{PRECLAIM\,(NBPayload* p) : void}{
  $verify\,$($p$)\;
  \eIf{p$\rightarrow$anti\_payload == NULL}{
    \uIf($p$ is not retired){p$\rightarrow$type != DELETE}{
      $PDELETE\,$($p$)\;
    } \uElseIf( $p$ is retired 2 epochs ago){p$\rightarrow$epoch < op\_epoch\,$-\,1$}{
      \Delete $p$\;
    } \uElse($p$ is retired less than 2 epochs ago){
      $to\_be\_freed\,$[$op\_epoch$].$add\,$($p$)\;
      \tcp*[l]{note to\_be\_freed[p$\rightarrow$epoch] may be unsafe}
    }
  }(p has an anti-payload attached.){
    $verify\,$($p\rightarrow anti\_payload$)\;
    \eIf($p$ is from 2 epochs ago){p$\rightarrow$epoch < op\_epoch$\,-\,1$}{
      \tcp*[l]{reclamation of anti-payload need to be deferred after a fence}
      $to\_be\_freed\,$[$p\rightarrow epoch$].$add\,$($p\rightarrow anti\_payload$)\;
      \Delete $p$\;
    }($p$ is from recent epoch){
      \tcp*[l]{equivalent to PDELETE}
      $to\_be\_freed\,$[$op\_epoch$\,+\,1].$add\,$($p\rightarrow anti\_payload$)\;
      $to\_be\_freed\,$[$op\_epoch$].$add\,$($p$)\;
    }
  }
}
\else
\BlankLine
\tcp*[l]{PRETIRE and PRECLAIM omitted due to space constraints}
\fi
\end{multicols}
\end{algorithm2e}
\caption{\Montage{} Pseudocode.}
\label{fig:montage-pseudocode}
\end{figure*}

Aside from the epoch clock itself, payloads are the only data allocated
in NVM.  Each payload indicates the epoch in which it was
created and whether it is new (\code{ALLOC}), a replacement of an
existing payload (\code{UPDATE}), or an anti-payload (\code{DELETE}).
\code{ALLOC} payloads are created only in \pcode{PNEW}.
\code{UPDATE} payloads are created only in \pcode{set} (when we discover
that the block being written was created in an earlier epoch, and cannot
be updated in place).
In lock-based data structures, \code{DELETE} payloads are created
only in \pcode{PDELETE}; they live
for exactly two epochs, until the payload they are nullifying can safely
be reclaimed. With nonblocking memory management, \code{DELETE} payloads
are created only in \pcode{PRETIRE}; they live for two epochs after
the corresponding \code{PRECLAIM} call.

\subsection{Storage Management}
\label{sec:Ralloc}

Space for payloads in \Montage{} is managed by a variant of the Ralloc
persistent allocator~\cite{cai-ismm-2020}.
Ralloc is in turn based on the nonblocking allocator of 
Leite and Rocha~\cite{leite2018lrmalloc}.
Ralloc has very low overhead and excellent locality during crash-free
operation.
Almost all metadata is kept in transient memory, and 
most allocation and deallocation operations perform no write-back or
fence instructions. 

In its original form, Ralloc performs garbage collection after a crash
to identify the blocks that are currently in use; all others are
returned to the free list.  For \Montage, we modified the recovery
mechanism to simply peruse all blocks, and to keep all and only those
that are labeled as having been created at least two epochs ago.  (These
blocks will of course have been written back at some previous epoch
boundary.)  \Montage{} passes the recovered blocks (i.e., payloads) to the
application data structure, which is then responsible for rebuilding
transient structures.  To facilitate parallel recovery, the application
can request that the blocks be returned via $k$ separate iterators, to
be used by $k$ separate application threads.

\subsection{Persistence, Epoch, and Reclamation Strategy}
\label{sec:design-space}

A wide variety of concrete designs could be used to flesh out
the pseudocode of Figure~\ref{fig:montage-pseudocode}.
Natural questions include
\begin{itemize}
\item
    Should the \code{advance\_epoch} function be called periodically by
    application (worker) threads---e.g., from within the API calls---or
    should it be called by a background thread?
\item
    Once \code{advance\_epoch} has been called, should it be executed by
    a single thread, or should it be parallelized?  (The Pronto system,
    a possible inspiration, can be configured to perform all writes-back
    on the sister hyperthread of the worker that wrote the
    data~\cite{memaripour-asplos-2020}.)
\item
    Is the answer to the previous question the same for both writes-back
    and storage reclamation?  Perhaps some tasks are better performed on
    the cores where payloads or payload lists are likely to be in cache?
\item
    Should all writes-back for a given epoch be delayed until the end,
    or does it make sense to start some of them earlier?  One
    might, for example, employ a circular buffer in each worker, and
    issue writes-back one at a time, all at once, or perhaps half a
    buffer at a time, as the buffer fills.
\item
    How long should an epoch be?  Should it be measured in time,
    operations performed, or payloads written?
\end{itemize}

We performed a series of experiments to evaluate the impact on
performance of various answers to these questions.  A summary of the
results appears in Section~\ref{sec:explore}.  The short answer is
that it seems to make sense, for the data structures we have explored to
date, and on the 2-socket Intel server in our lab, to have a single
background thread that is responsible for all the work of
\code{advance\_epoch}, and to have it perform this work every
10--100\,ms.

\section{Experimental Results}
\label{sec:results}

In this section, we report experiments on queues, maps,
graphs, and memcached to evaluate \Montage's performance and generality,
and to answer the following questions:
\begin{itemize}
\item
  What is the best way to configure \Montage?
  (Sect.~\ref{sec:explore})
\item
  How does \Montage{} compare to prior special- and general-purpose
  systems, and to baseline transient structures?
  (Secs.~\ref{sec:microbenchmarks}--\ref{sec:exp_graph})
\item
  What is the cost of recovery? (Sec.~\ref{sec:exp_rec})
\end{itemize}

\subsection{Hardware and Software Platform}
\label{sec:platform}

All tests were conducted on a Linux 5.3.7 (Fedora 30) server with two
Intel Xeon Gold 6230 processors, with 20 physical cores and 40
hyperthreads in each socket---a total of 80 hyperthreads.  Threads in
all experiments were pinned first one per core on socket 0, then on the
extra hyperthreads of that socket 0, and then on the second socket.
Each socket has 6 channels of 128\,GB Optane DIMMs and 6 channels of
32\,GB DRAMs. We use ext4 to map NVM pages in direct access (DAX) mode.
\ifverbose\else
In all experiments, we allow Linux to allocate DRAM across the two
sockets of the machine according to its default policy.  The NVM is
explicitly interleaved across sockets (\code{dm-stripe} with a
2\,MB chunk size~\cite{scargall-dmstrip-2018}).
\fi
The source code of Montage is available at
https://github.com/urcs-sync/Montage.

Systems and structures tested include the following: 
\begin{description}
\item[\Montage{} --] as described in previous sections.
\item[Friedman --] the persistent lock-free queue of Friedman
  et al.~\cite{friedman-ppopp-2018}.
\item[Dal\'{\i} --] our reimplementation of the buffered durably
  linearizable hash table of Nawab et al.~\cite{nawab-disc-2017}.
\item[SOFT --] the lock-free hash table of Zuriel et
  al.~\cite{zuriel-oopsla-2019}, which persists only semantic data
  but keeps a full copy in DRAM.
\item[MOD --] persistent structures (here, queues and hash maps) as
  proposed by Haria et al.~\cite{haria-asplos-2020}, who leverage
  history-preserving trees to linearize updates with
  a single write.
\item[Pronto-Full {\rm and} Pronto-Sync --] the general-purpose system
  of Memaripour et al.~\cite{memaripour-asplos-2020}, which logs
  high-level operation descriptions that can be replayed, starting
  from a checkpoint, to recover after a crash.
  We test both the synchronously logged and (on $\leq$ 40
  threads) the ``full'' (asynchronous) version.
\item[Mnemosyne --] the general-purpose, pioneering system of
  \mbox{Volos} et al.~\cite{volos-asplos-2011}, which adds persistence
  to the TinySTM transactional memory system~\cite{riegel-disc-2006}.
\end{description}

For comparison purposes, we also include:
\begin{description}
\item[DRAM\,(T) {\rm and} NVM\,(T) --] high quality transient data
  structures built on DRAM and NVM, respectively, with no persistence
  support.
\item[\Montage\,(T) --] a variant of \Montage{} that still places payloads in
  NVM, but elides all persistence operations (no buffering,
  write-back instructions, delayed deletion, or epoch advance).
\end{description}

\subsection{Sensitivity to Design Alternatives}
\label{sec:explore}

As noted in Section~\ref{sec:design-space}, there is a very large design
space for the outline given in Figure~\ref{fig:montage-pseudocode}.
In Figure~\ref{fig:design-explore}, we show the performance of a
\Montage{} hash map across several design dimensions.
Each narrow bar indicates total throughput for 40 threads
running on a single processor of our test machine.  Each thread performs
lookup, insert, and delete operations in a 2:1:1 ratio.
The table has one million buckets.  It begins half full and remains so,
as keys are drawn from a million-element range.
In each group of bars, the epoch length varies from 1\,$\mu$s to 15s.
We used time to measure epoch length because it does not vary across
threads.

We consider three main strategies for write-back.
In \emph{DirectWB}, each update operation initiates the
write-back of its payloads immediately after completing.  This strategy
is somewhat risky on current Intel processors: the \code{clwb} (cache
line write-back) instruction actually evicts the line, raising the
possibility of an unnecessary subsequent last-level cache miss if the
line is still in the working set.  (Future processors are expected to
retain the line in shared mode.)
In \emph{PerEpoch}, each thread maintains a buffer in which
it stores the address and length of every written payload.  These are
saved for write-back at the end of the next epoch.
In the intermediate \emph{BufferedWB} strategy, each buffer holds only
1000 entries, half of which are written back when that capacity is reached.

\begin{figure}
  \includegraphics[width=0.475\textwidth]{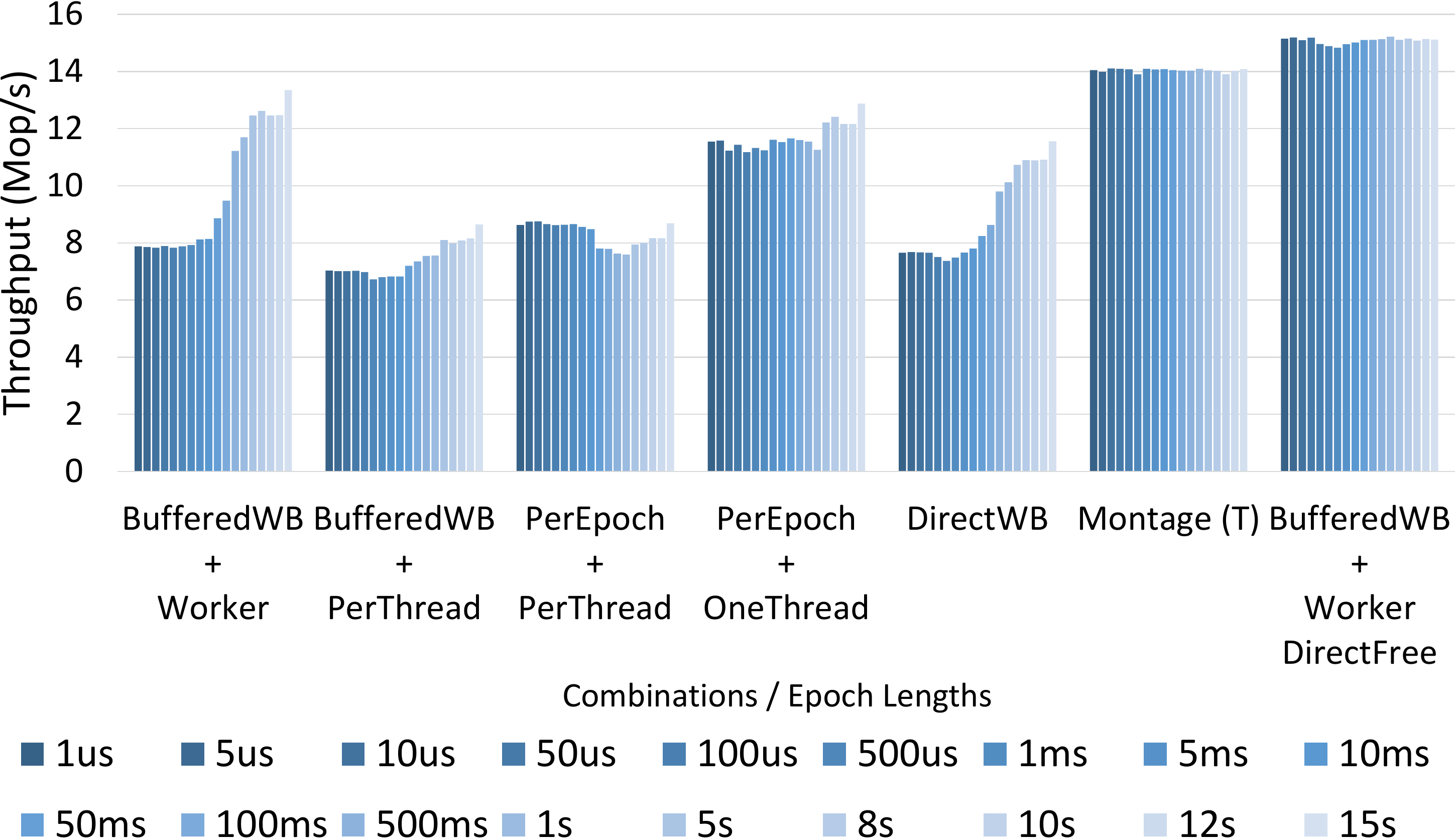}
\caption{Hash map throughput as a function of write-back strategy,
reclamation strategy, and epoch length.}
\label{fig:design-explore}
\end{figure}

For \emph{BufferedWB}, we tried letting each worker
perform its own writes-back (\emph{Worker}), or arranging for these
to happen on the sister hyperthread, which shares
the L1 cache (\emph{PerThread}).  For \emph{PerEpoch}, we tried performing
the writes-back on the sister hyperthreads (\emph{PerThread}) or in a
single background thread that empties all the buffers (\emph{OneThread}).

\newlength{\microfigwidth}
\begin{figure*}
  \centering
  \microfigwidth .4\textwidth
  \strut\hfill
  \subfloat[Concurrent Queues\label{fig:queue_30s_strip}]%
      {\includegraphics[width=1.17\microfigwidth]%
          {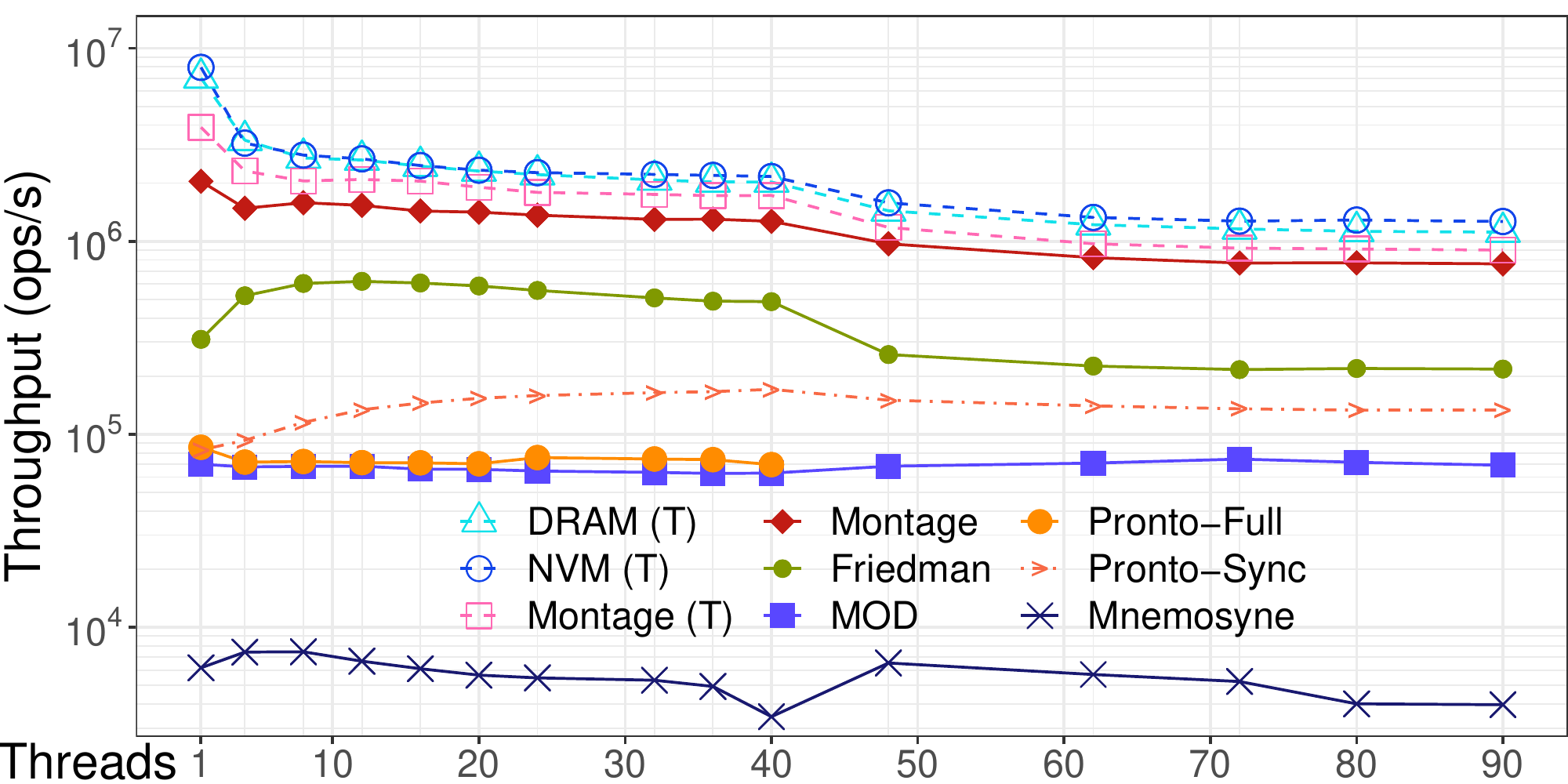}%
      }
  \hfill
  \subfloat[Concurrent Hash Tables---0:1:1
  get:insert:remove\label{fig:hashtables_g0i50r50_strip}]%
      {\includegraphics[width=1.17\microfigwidth]%
          {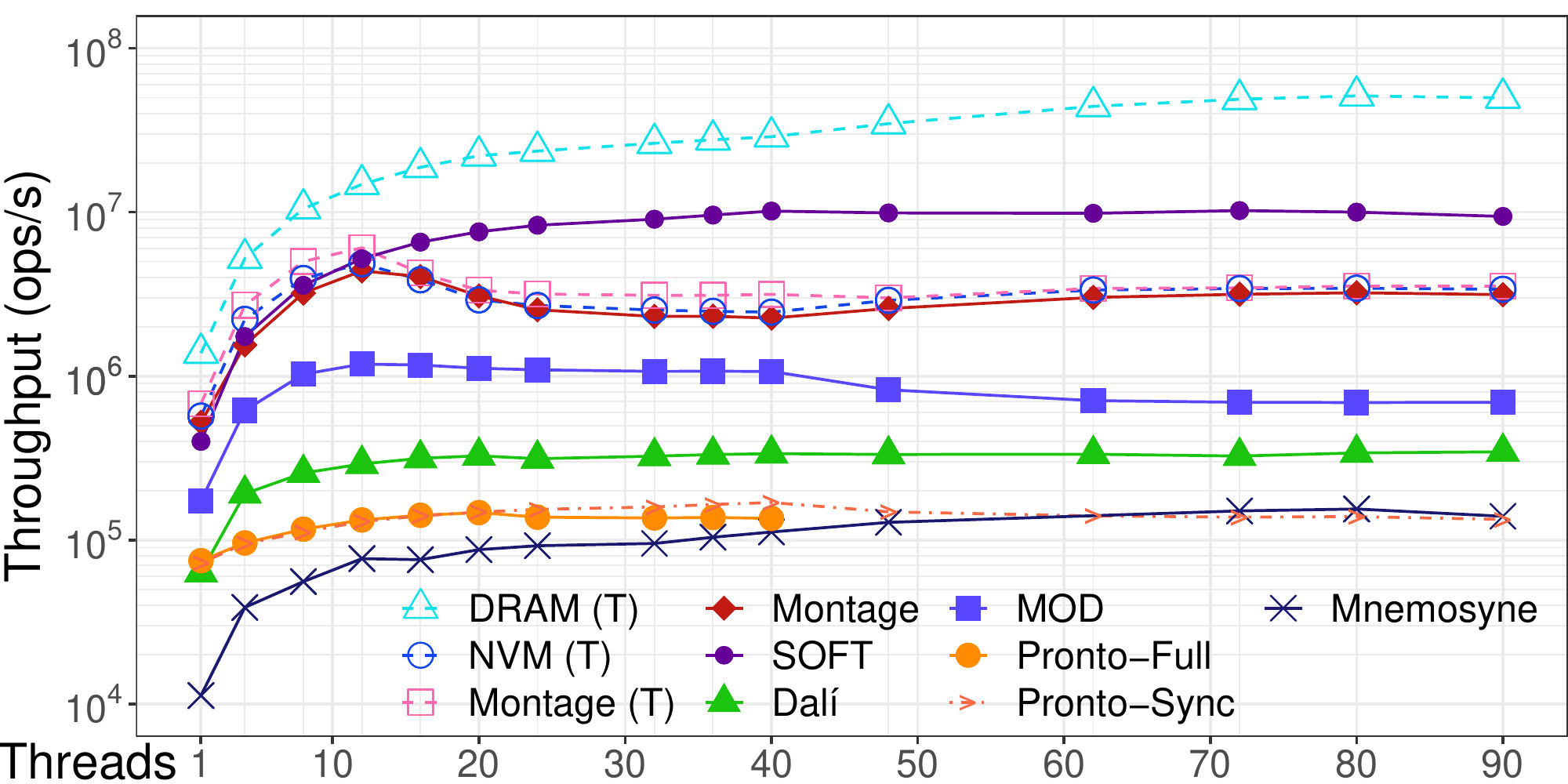}%
      }
  \hfill\strut\\
  \vspace{1.5ex}
  \strut\hfill
  \subfloat[Concurrent Hash Tables---2:1:1 get:insert:remove
  \label{fig:hashtables_g50i25r25_strip}]%
      {\includegraphics[width=1.17\microfigwidth]%
          {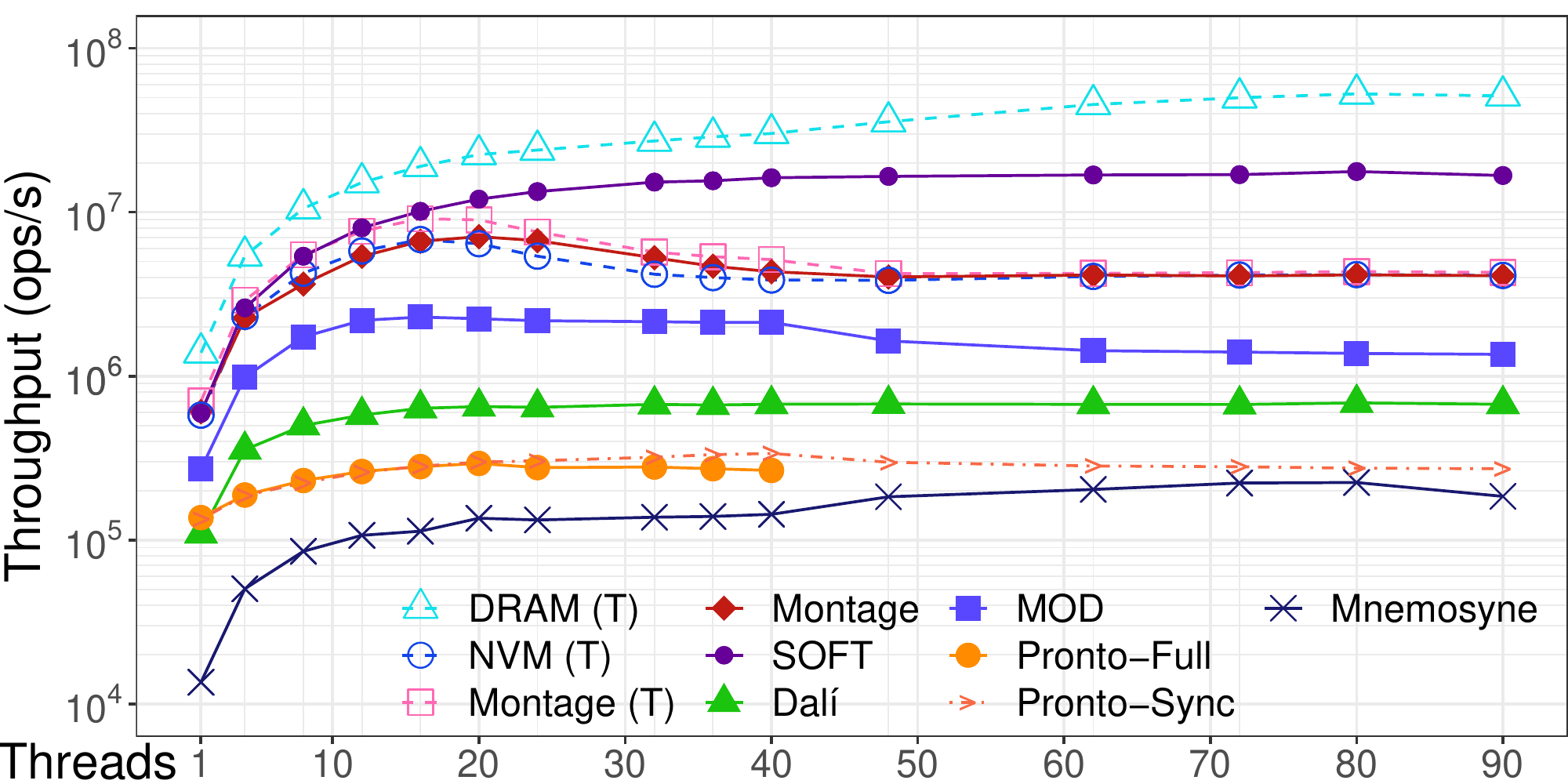}%
      }
  \hfill
  \subfloat[Concurrent Hash Tables---18:1:1
  get:insert:remove\label{fig:hashtables_g90i5r5_strip}]%
      {\includegraphics[width=1.17\microfigwidth]%
          {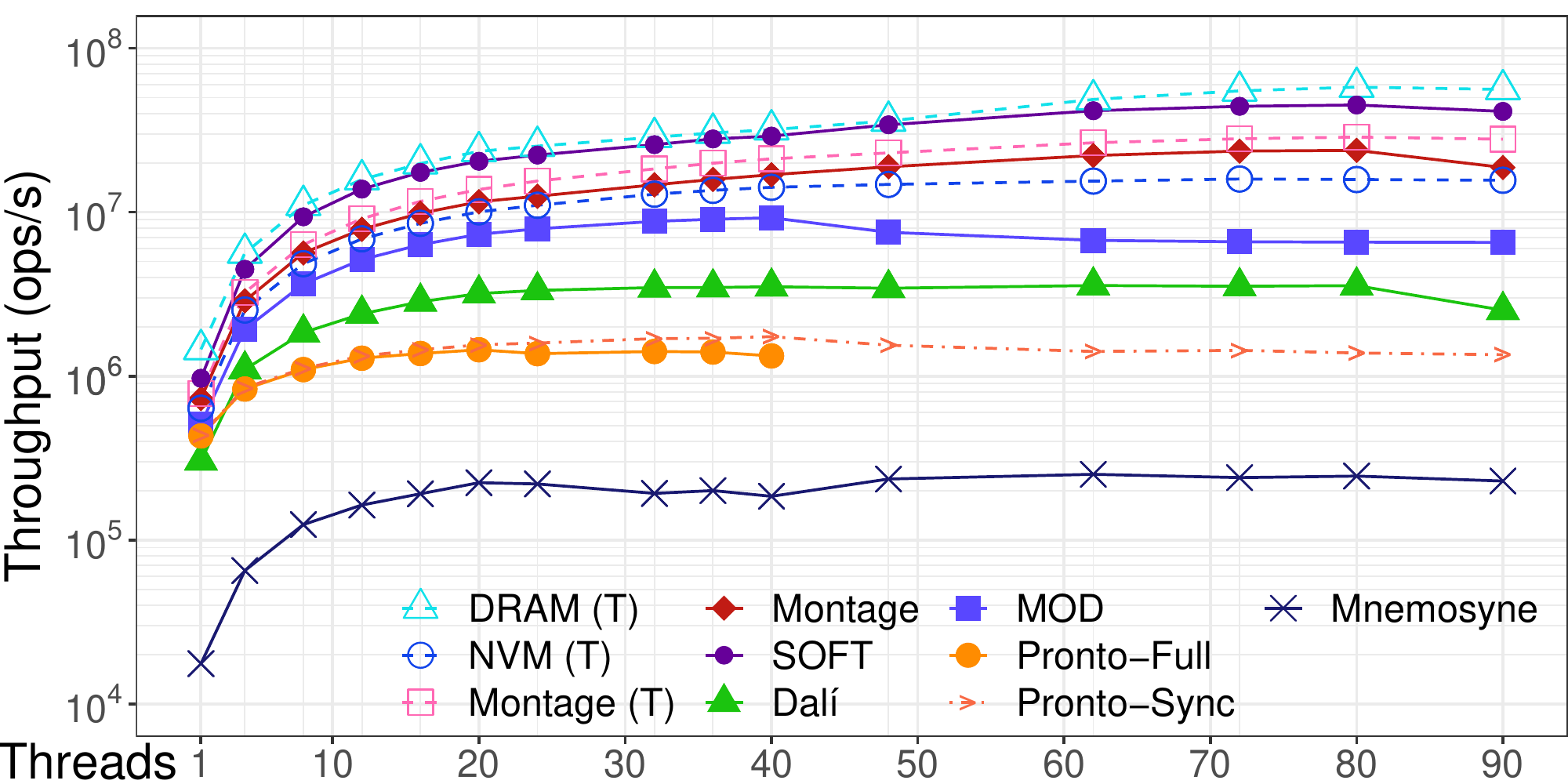}%
      }
  \hfill\strut
  \caption{Throughput of Concurrent Data
      Structures\ifverbose { on Interleaved NVM}\fi.}
  \label{fig:concurrent-interleaved}
\end{figure*}

\ifverbose
\begin{figure*}
  \centering
  \microfigwidth .4\textwidth
  \strut\hfill
  \subfloat[Concurrent Hash Tables---2:1:1
  get:insert:remove\label{fig:hashtables_g50i25r25}]
      {\includegraphics[width=1.17\microfigwidth]%
          {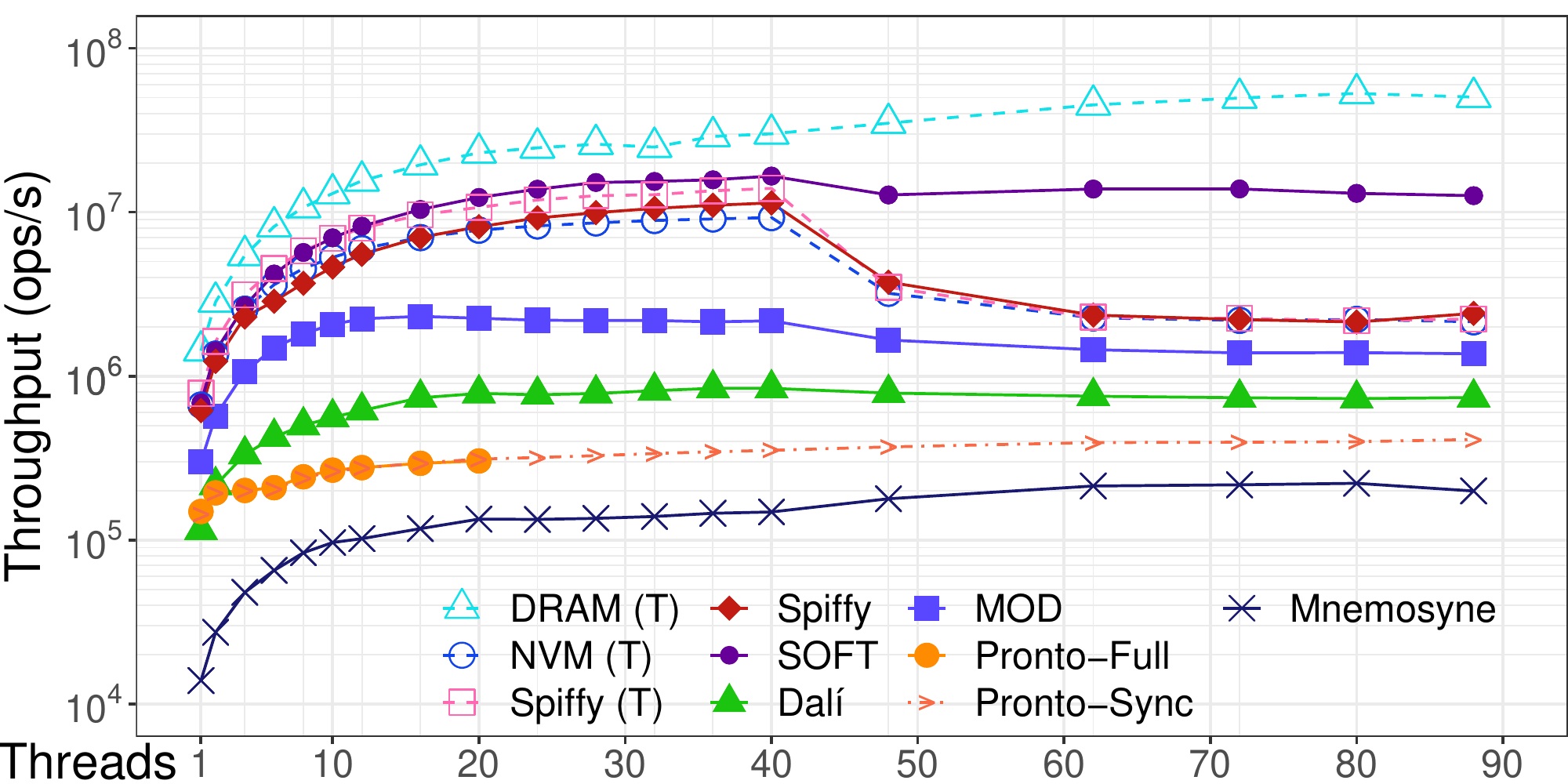}%
      }
  \hfill
  \subfloat[Concurrent Hash Tables---18:1:1
  get:insert:remove\label{fig:hashtables_g90i5r5}]
      {\includegraphics[width=1.17\microfigwidth]%
          {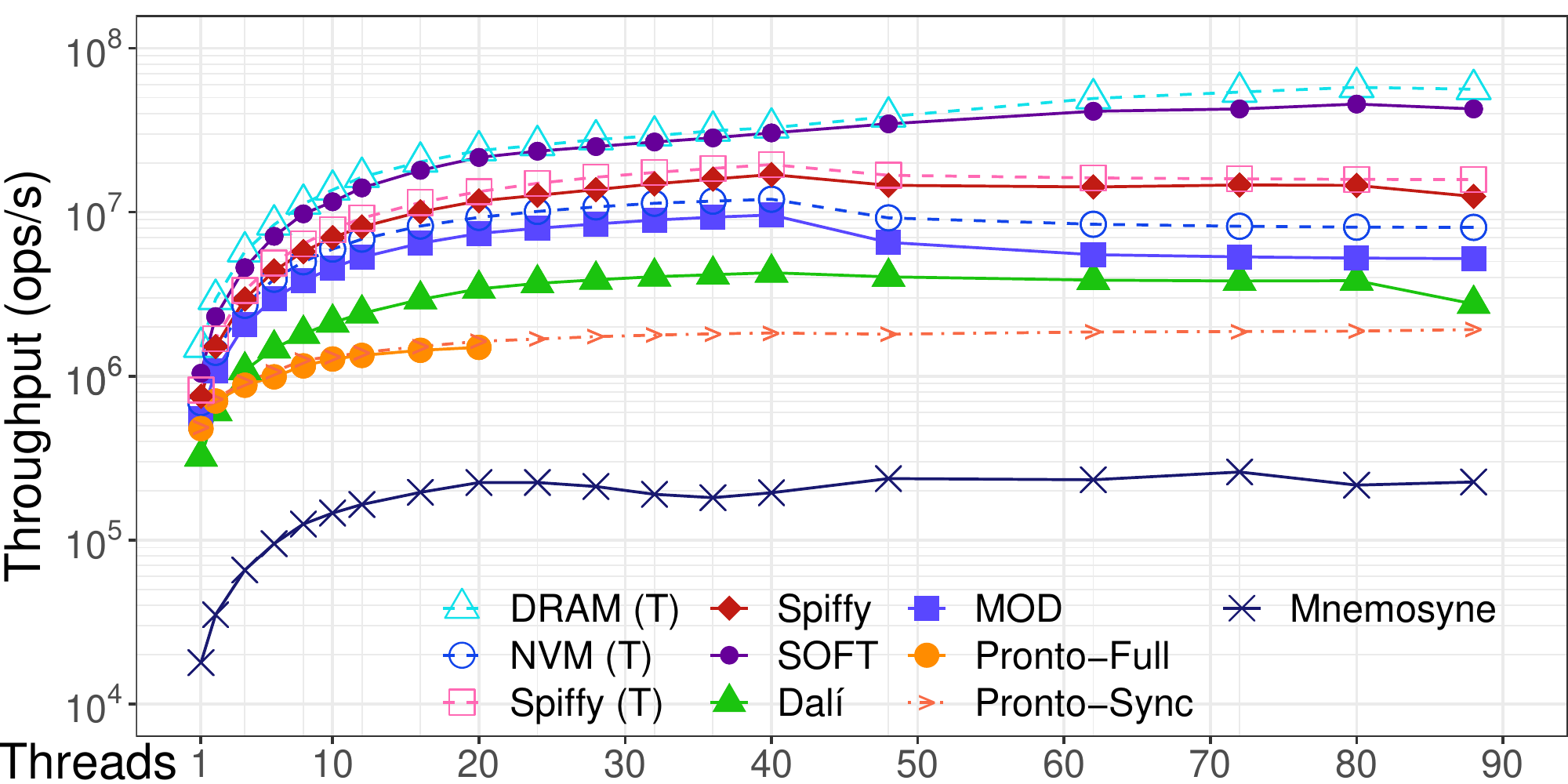}%
      }
  \hfill\strut
  \caption{Throughput of Concurrent Hash Tables on First Socket's NVM.}
  \label{fig:concurrent-socket1}
\end{figure*}
\fi

Memory reclamation always has to be delayed for two epochs, as explained
in Section~\ref{sec:epochs}.
We experimented with strategies in which payloads were always reclaimed
on the core on which \code{PDELETE} was called, but any benefits of
locality were outweighed by the cost of global synchronization.
Fortunately, even though it relies on thread-local free lists, Ralloc is
able to rebalance these very efficiently; it imposes very little penalty
for deallocating all blocks in the epoch-advancing thread.
For comparison purposes, we also ran experiments \Montage\,(T) and
with an (unsafe) variant of \emph{BufferedWB} that does not delay
reclamation.

As shown in Figure~\ref{fig:design-explore},
\emph{OneThread} generally beats \emph{Worker} and
\emph{PerThread}, particularly given that the latter requires so
many extra pipeline resources.
\emph{PerEpoch} generally beats \emph{BufferedWB};
this may be due to the impact of working set eviction or to the fact
that, by the end of the epoch, many payloads have already been written
back implicitly.
Overall, \emph{PerEpoch}+\lb{}\emph{OneThread} is clearly the
winning strategy.
A single thread turns out to have ample bandwidth\ifverbose { for the
task}\fi, even given the cache misses it suffers while perusing
per-thread buffers.

Note that most groups of bars slope upward with very long epochs.
Given our workload, when epochs are measured in seconds, most
existing payloads are deleted early in the epoch; subsequent
deletions tend to be ``same-epoch'' payloads, and can be reclaimed
without delay.
This is verified both by the fact that ``\Montage\,(T)'' has no curve to
its bar group and by a follow-up experiment showing that the rise in
performance comes later when the key range is larger.
The closeness in height of the final two bar groups suggests that most of
\Montage's overhead is due to memory management, not to writes-back.
In \emph{PerThread} and \emph{DirectWB}, the benefit of very short
epochs, relative to medium-length, presumably reflects the fact that
buffers
are likely to remain in the L1 or L2 cache if epochs are extremely short.

In other experiments (not shown), we varied several additional
parameters, including buffer sizes and emptying fractions for
\emph{BufferedWB}, and synchronization mechanisms for \emph{PerThread}.
None of these produced significantly different
results.  For the remaining experiments in this paper, \Montage{} is
configured with per-thread, whole-epoch buffers, an epoch length of
50\,ms, and a single background thread responsible for epoch advance,
write-back, and memory reclamation.

\subsection{Performance Relative to Competing Systems}
\label{sec:microbenchmarks}

We have benchmarked \Montage{} against the data structures and systems
listed in Section~\ref{sec:platform}, using queue and hash map
structures.  Results appear in Figure~\ref{fig:concurrent-interleaved}.
The \Montage{} queue employs a single lock; the \Montage{} hash map has
a lock per bucket.
In work not reported here, we have developed 
nonblocking linked lists, queues, and maps, and various tree-based maps.
In Section~\ref{sec:exp_graph} we describe the
implementation of a general graph, with operations to add, remove, and
update vertices and edges.

The queue microbenchmark runs a 1:1 enqueue:dequeue workload. For the map
we run three different workloads---write-\lb{}dominant (0:1:1
get:\lb{}insert:\lb{}remove), read-write (2:1:1
get:\lb{}insert:\lb{}remove), and read-dominant (18:1:1
get:\lb{}insert:\lb{}remove),\footnote{%
  SOFT does not support atomic \code{replace} or \code{update},
  so the benchmark does not include
  them.  In separate experiments (not shown here), we confirmed that
  \code{update} does not change the curves for other
  algorithms.}
with 0.5 million elements preloaded in 1 million hash buckets. The
value size in queues and maps is 1\,KB. The key in maps ranges from 1
to 1 million, converted to string and padded to 32\,B. The benchmarks
run between 1 and 90 threads.  Each workload runs for 30 seconds.
Results were averaged over 3 trials for each data point.

As shown in Figure~\ref{fig:concurrent-interleaved}, \Montage{} data
structures generally perform as fast as transient structures
running on NVM (they may even outperform NVM\,(T),
given transient indexing in DRAM).  Compared to DRAM\,(T), \Montage{}
adds as little as 30\% overhead in queues, and less than
an order of magnitude on the highly concurrent
hash table (less than
70\% at low thread counts).
With the exception of SOFT, \Montage{} also outperforms all tested
persistence systems on all four workloads.
The \Montage{} queue provides more than $2\times$ the throughput of
Friedman et al.'s
\ifverbose special-purpose \fi
queue, and is more than an order of magnitude faster
than the MOD, Pronto, and Mnemosyne queues.
For hash maps, \Montage{} runs more than $2\times$ faster
than MOD,\footnote{%
  We implemented a per-bucket locking hash table using MOD linked
  lists. This hash table has lower time complexity and better
  scalability than the compressed hash-array mapped prefix-tree
  in the original MOD paper~\cite{haria-asplos-2020}.}
$4\times$--$30\times$ faster than Dal\'{\i} and Pronto,
and nearly two orders of magnitude faster than Mnemosyne on the
write-dominant and
read-write workloads. On the read-dominant workload, \Montage{}
still has around $2\times$ the throughput of MOD at most thread
counts.

The exceptional case is SOFT, which maintains---and reads
from---a full copy of the data in DRAM\@. Nonetheless, \Montage{} is
close or outperforms SOFT at low thread counts and on the read-dominant
workload, and still achieves more than $1/3$ the throughput of SOFT at
high thread counts.  The downside is that by keeping a full
copy in DRAM, SOFT loses the ability to take full advantage of the
10$\times$ additional capacity of NVM\@.
Interestingly, \Montage{} and NVM\,(T) stop scaling at 12 and
20 threads on the write-dominant and read-write workloads, which may
reflect multithreading contention in NVM's write combining buffer and write pending queues~\cite{yang-fast-2020}.

\ifverbose
\subsubsection*{Memory Configuration}

In all our experiments, we allow Linux to allocate DRAM across the two
sockets of the machine according to its default policy.  NVM, however,
must be manually configured.  In the experiments of
Figure~\ref{fig:concurrent-interleaved}, we interleaved it across
the sockets (\code{dm-stripe} with a 2\,MB chunk
size)~\cite{scargall-dmstrip-2018}.  In separate experiments, we
configured the sockets as separate domains, and placed all \Montage's
payloads on socket~0.  With threads pinned as before (numbers 40 and up
on socket 1), results for the read-write and read-dominant hash maps
are shown in Figure~\ref{fig:concurrent-socket1}.
With 50\% writes, \Montage{} scales better before 40 threads on the
first socket's NVM than on interleaved NVM, but drops
heavily once threads cross sockets.  NVM\,(T) is also affected.
A possible explanation to the drop: frequent, slow remote writes
occupy NVM bandwidth and block other accesses~\cite{yang-fast-2020}.
This cost is amortized and mitigated by the doubled number of DIMMs if
NVM is interleaved across the sockets.
Consistent to the explanation, the read-dominant workload
behaves more smoothly while crossing sockets. The graph for the write-dominant
workload (not shown here) is similar to
Figure~\ref{fig:hashtables_g50i25r25}.
Queues
show no significant difference from the interleaved case.

\fi

\subsubsection*{Payload Size}
\label{sec:exp_single}

\begin{figure*}
  \centering
  \microfigwidth .4\textwidth
  \strut\hfill
  \subfloat[Single-threaded Queues\label{fig:queue_30s_size}]%
      {\includegraphics[width=1.17\microfigwidth]%
          {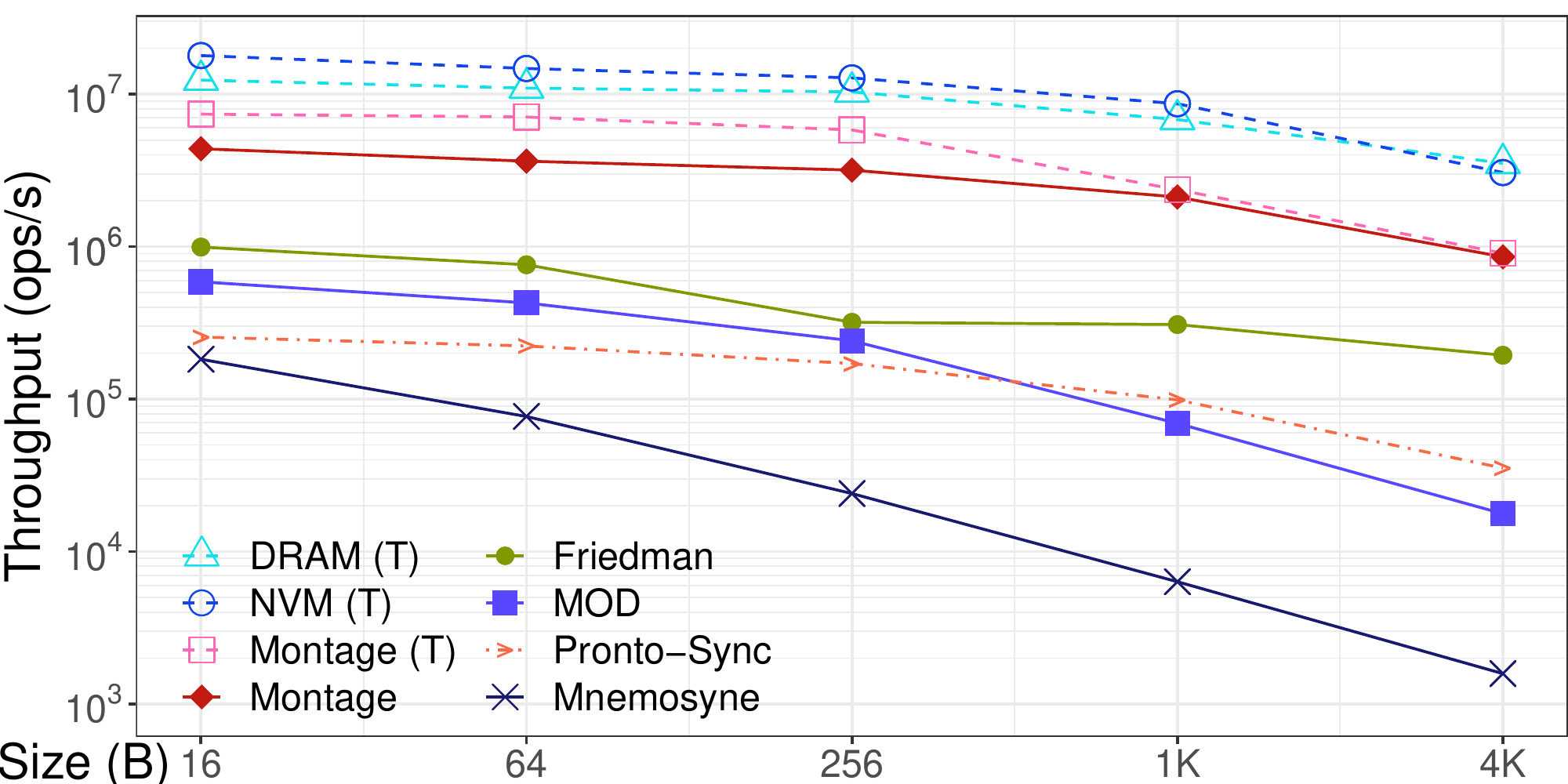}%
      }
  \hfill
  \subfloat[Single-threaded Hash Tables---2:1:1
  get:insert:remove\label{fig:hashtables_g0i50r50_size}]
      {\includegraphics[width=1.17\microfigwidth]%
          {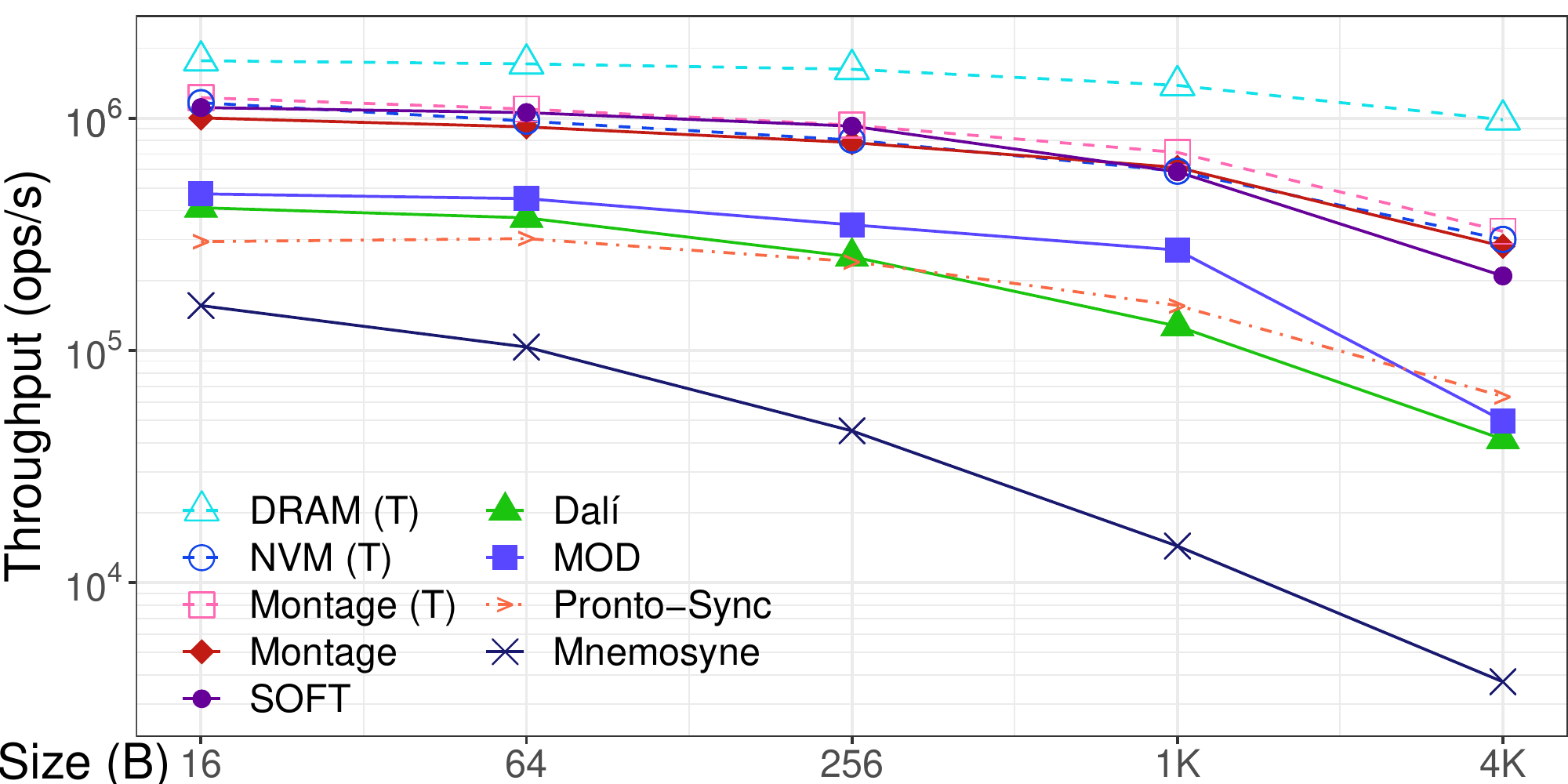}%
      }
  \hfill\strut
  \caption{Throughput of Single-threaded Data Structures.}
  \label{fig:single}
\end{figure*}

To assess the impact of operation footprint on relative performance,
we repeated our queue and read-write hash table experiments with
a single thread but with payloads varying from 
16\,B to 4\,KB.
Results appear in Figure~\ref{fig:single}.

As in the previous section, \Montage{} outperforms all competitors but
SOFT\@.  This is also the case on write-dominant and read-dominant
workloads (not shown). Interestingly, in the write-heavy
case, the SOFT curve drops more sharply than the \Montage{} curve, and
crosses over at just 256\,B: the overhead of (strict) durable
linearizability increases with larger payloads, while \Montage{}
benefits more from its buffering.

\subsection{Hash Map Validation Using memcached}
\label{sec:memcached}

To confirm the data structure results in a more realistic application,
we use \Montage{} to persist a variant of \emph{memcached} developed by
Kjell\-qvist et al.~\cite{kjellqvist-icpp-2020}.  This variant links
directly to a multithreaded client application, dispensing with the
usual socket-based communication.  It was appealing for our experiments
because the authors had already converted it to use Ralloc instead of
the benchmark's own custom allocator.

\begin{figure*}
\centering
\begin{minipage}[t]{.65\textwidth}
  \captionsetup{justification=centering}
  \includegraphics[width=.48\textwidth]{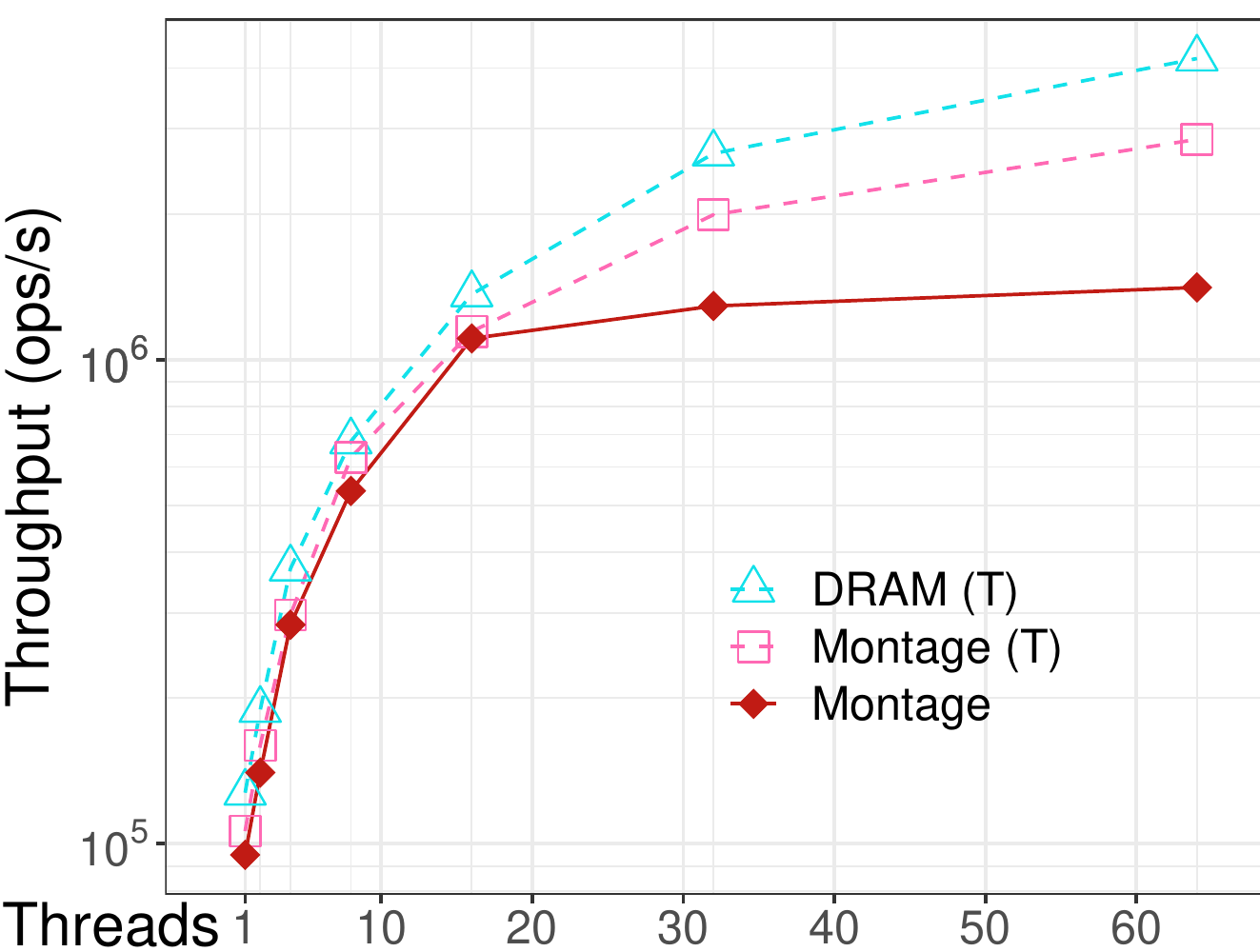}
  \hfill
  \includegraphics[width=.48\textwidth]{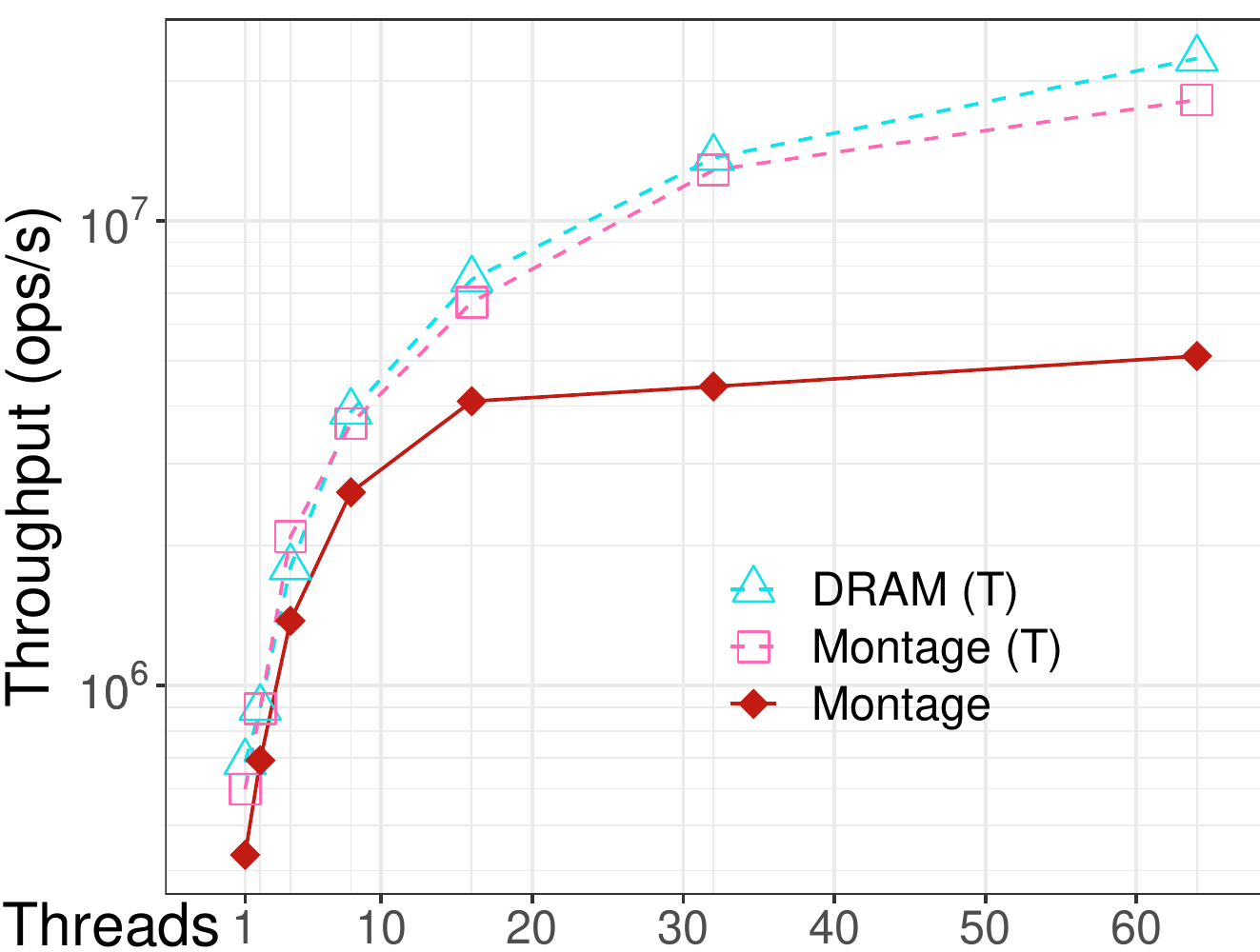}
  \caption{Graph Microbenchmark Throughput
  (note difference in y-axis scales)\\
      AddEdge:GetEdge:RemoveEdge:ClearVertex =
      1:1:1:1 (left), 33:33:33:1 (right).}
  \label{fig:graph_fair_vs_unfair}
\end{minipage}
\hfill
\begin{minipage}[t]{.33\textwidth}
  \includegraphics[width=\textwidth, height=1.65in]{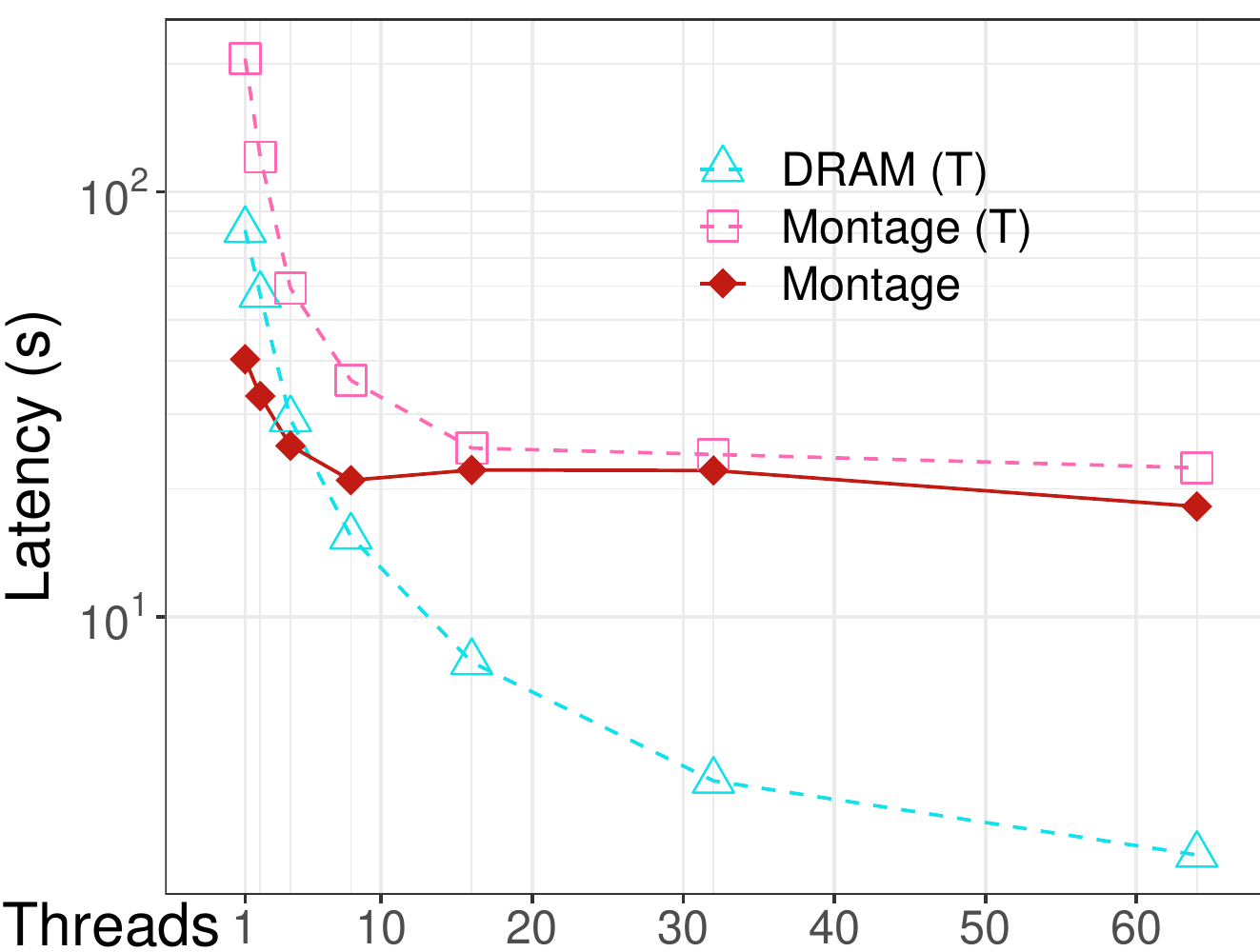}
  \caption{Time to Rebuild Orkut Graph.}
  \label{fig:graph_Orkut}
  \end{minipage}
\end{figure*}

\begin{figure}
  \includegraphics[width=.96\columnwidth]%
      {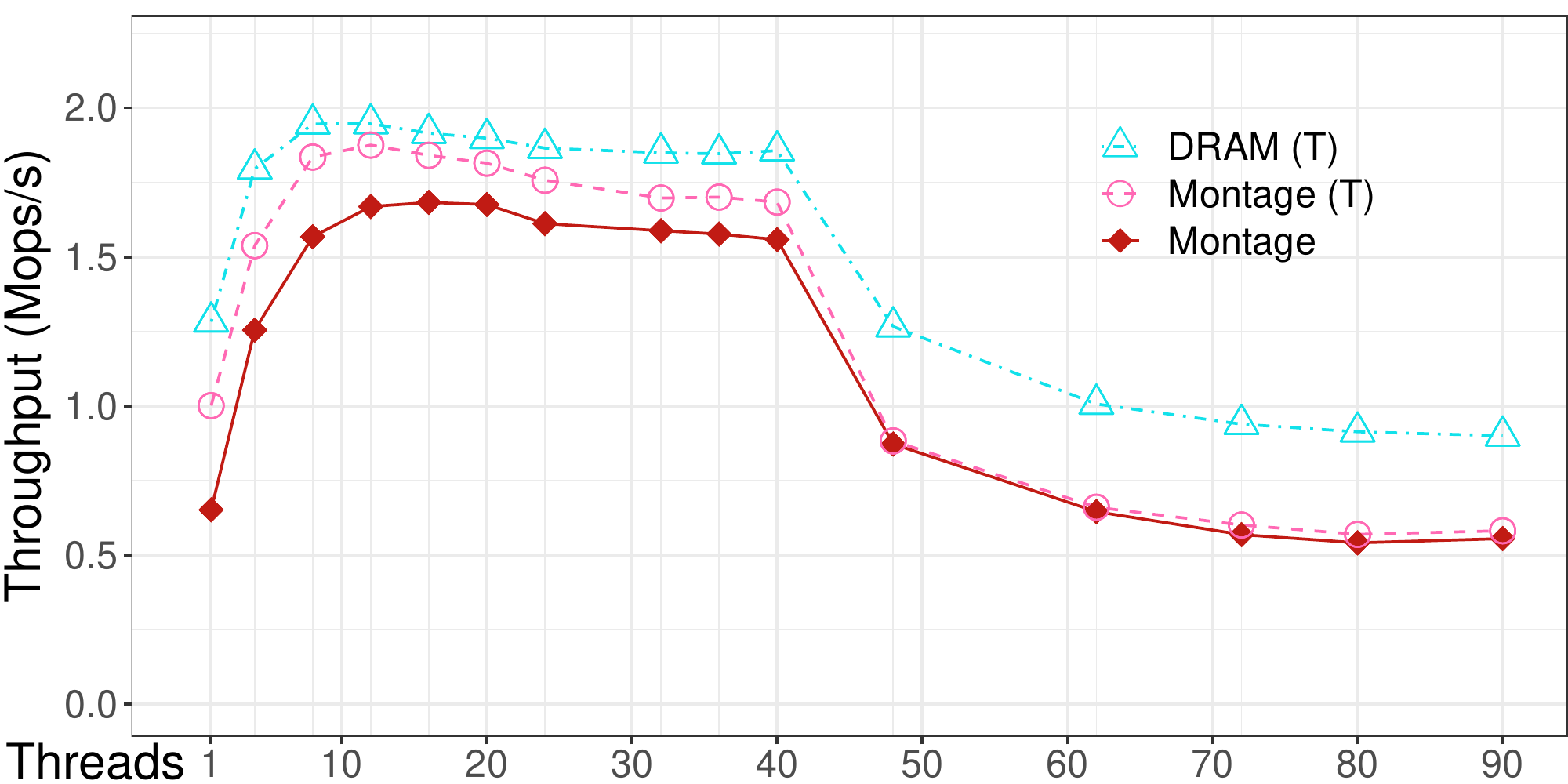}
  \caption{Throughput of memcached on YCSB-A.}
  \label{fig:threadcached_ycsba}
\end{figure}
Figure~\ref{fig:threadcached_ycsba} compares
\ifverbose the performance of \fi
the resulting (fully persistent, recoverable) version of memcached to the
transient version of Kjell\-qvist et al., placing items on DRAM or on NVM\@.
Here the YCSB-A workload~\cite{cooper-ycsb-2010}, running on 1 million
records, consists of 2.5 million read and 2.5 million update operations,
evenly distributed to each thread.  Data points reflect the average of
three trials.
\ifverbose
The results shown are for interleaved NVM; results with
all data on socket 0 were very similar.
\fi
As in the microbenchmark results, \Montage{} performs within a small
constant factor of purely transient structures.

\subsection{Generality in Graphs}
\label{sec:exp_graph}

As noted in Section~\ref{sec:api}, it is important in \Montage{} to
avoid long chains of pointers.  To build a persistent graph, we
therefore arrange for edge payloads to point to their endpoint vertices,
but not vice versa.  A more conventional representation of
connectivity is then kept in a transient structure, with the (typically
large) edge and vertex attributes appearing only in NVM payloads.
We regard the feasibility of building a graph in \Montage{} as a strong
indication of the system's generality.

Using our \Montage{} graphs, we compare performance
(as in the memcached experiments) to transient graphs placed in DRAM or
NVM\@.  Figure~\ref{fig:graph_fair_vs_unfair} shows results for a
microbenchmark that performs a mix of
\code{AddEdge}, \code{RemoveEdge}, \code{GetEdge}, and
\code{ClearVertex} operations. The first three of these take two vertex
IDs as source and destination; the fourth deletes a vertex and removes
all its in- and out-edges.

To initialize the graph, for each vertex $v
\in V$, we sample $n \sim N(\mu=10,\sigma=3)$ and create $n$ edges
$(v,v')$, with $v' \sim U(0,|V|-1)$.
By maintaining a small average vertex degree, this approach
avoids very large atomic operations.  Indeed, when
\code{ClearVertex} is called less often (right half of
Fig.~\ref{fig:graph_fair_vs_unfair}), overall throughput is higher
but the constant, per-operation component of \Montage's overhead has a
relatively higher impact, and the gap between \Montage{} and the transient
structures is larger.

\subsection{Recovery Time}
\label{sec:exp_rec}

To assess the overhead of recovery in \Montage, we measured both hash
map and graph examples.
In the hash map case, we initialized the table with 2--64 million
1\,KB elements, leading to a total payload size of 1--32\,GB.
With 1 recovery thread, \Montage{} recovers the 1\,GB data set in 0.7\,s
and the 32\,GB data set in 41.9\,s.
With 8 recovery threads, it takes 0.4 and 13.8\,s, respectively.
Improving the scalability of recovery is a topic for future work.

As a second example, we compared the recovery time of a large \Montage{}
graph
(the SNAP Orkut dataset~\cite{leskovec2016snap,
snapOrkut}, a social network of $\sim$3\,M vertices and 117\,M
edges)
to the time required to construct the same graph from a file of
adjacency lists.
The dataset is partitioned into many files, each of which uses
a custom binary format that
eliminates the need for string manipulation.
\Montage{} recovery is handled much like the parallel I/O: vertices and
edges are added back to the graph in parallel. Because recovery is an
internal graph operation, however, much of the locking can be elided by
cyclically distributing vertices among threads, each of which
creates a set of edge buffers to pass to other threads.
Figure~\ref{fig:graph_Orkut} demonstrates that recovery is even faster than
reconstruction on DRAM at low thread counts, and takes roughly as long as
reconstruction on NVM after 16 threads.
Crucially, the \Montage{} implementation has the advantage of supporting
small changes to the graph without the need orchestrate persistence via
file I/O.

\section{Conclusions}
\label{sec:conclusions}
We have introduced \Montage, the first general-purpose system for
buffered durable linearizability of persistent data structures.
In comparison to systems that are (strictly) durably linearizable,
\Montage{} moves write-back and, crucially, fencing off the critical path
of the application.
\Montage{} is built on top of the Ralloc nonblocking persistent
allocator~\cite{cai-ismm-2020},
which avoids both writes-back and fences in most allocation and
deallocation operations.
Nonblocking data structures remain nonblocking when implemented on top
of \Montage, though preempted threads can stall the advance of the
persistence frontier.

Experiments with multiple data structures---including memcached's hash
table---confirm that \Montage{} dramatically
outperforms prior general-purpose systems for persistence.
It also outperforms---or is competitive with---existing special-purpose
persistent data structures.
In many cases, in fact, it rivals the performance of traditional
transient data structures configured to use NVM instead of DRAM\@.  This
is generally the best performance one could hope for.

\ifverbose
In most of our experiments, \Montage{} was used to persist data
structures used by a single multithreaded application at a time.  Our
experiments with memcached, however, leveraged code developed for the
Hodor project~\cite{hedayati-atc-2019}, which allows a data structure to
be shared safely among mutually untrusting applications, with
independent failure modes.  We speculate that such structures may
provide a particularly attractive alternative to files for shared
abstractions in a future filled with nonvolatile main memory.
\fi

\bibliographystyle{abbrv}
\bibliography{main}

\end{document}